\documentclass[11pt,letterpaper]{article}
\usepackage{fullpage}
\usepackage[hidelinks]{hyperref}
\usepackage{graphicx}
\usepackage{preamble}
\usepackage{pgfplots}
\usepackage{natbib}
\usepackage{url}
\pgfplotsset{compat=1.18}

\title{The Distortion of Prior-Independent $b$-Matching Mechanisms}

\author{Ioannis Caragiannis\thanks{Department of Computer Science, Aarhus University, {\AA}bogade 34, 8200 Aarhus N, Denmark. Email: \protect\url{iannis@cs.au.dk}} \and Vasilis Gkatzelis\thanks{Department of Computer Science, Drexel University, 3675 Market Street, Philadelphia, PA 19104, USA. Email: \protect\url{gkatz@drexel.edu}} \and Sebastian Homrighausen\thanks{Department of Computer Science, Aarhus University, {\AA}bogade 34, 8200 Aarhus N, Denmark. Email: \protect\url{homrighausen@cs.au.dk}}}

\date{}

\begin{document}

    \maketitle

\begin{abstract}
In a setting where $m$ items need to be partitioned among $n$ agents, we evaluate the performance of mechanisms that take as input each agent's \emph{ordinal preferences}, i.e., their ranking of the items from most- to least-preferred. The standard measure for evaluating ordinal mechanisms is the \emph{distortion}, and the vast majority of the literature on distortion has focused on worst-case analysis, leading to some overly pessimistic results. We instead evaluate the distortion of mechanisms with respect to their expected performance when the agents' preferences are generated stochastically. We first show that no ordinal mechanism can achieve a distortion better than $e/(e-1)\approx 1.582$, even if each agent needs to receive exactly one item (i.e., $m=n$) and every agent's values for different items are drawn i.i.d.\ from the same known distribution. We then complement this negative result by proposing an ordinal mechanism that achieves the optimal distortion of $e/(e-1)$ even if each agent's values are drawn from an agent-specific distribution that is unknown to the mechanism. To further refine our analysis, we also optimize the \emph{distortion gap}, i.e., the extent to which an ordinal mechanism approximates the optimal distortion possible for the instance at hand, and we propose a mechanism with a near-optimal distortion gap of $1.076$. Finally, we also evaluate the distortion and distortion gap of simple mechanisms that have a one-pass structure. 
\end{abstract}

\section{Introduction}

We consider a class of $\mathbf{b}$-matching problems where $n$ agents compete over $m\geq n$ indivisible items and each agent $i$ has a quota $b_i$ (the maximum number of items they can receive), such that $\sum_{i=1}^n b_i =m$. This framework generalizes the classic one-to-one matching problem, where $b_i=1$ for all $i$ \citep{hylland1979efficient,zhou1990conjecture,bogomolnaia2001new}, and captures a wide range of real-world applications, including the allocation of computational resources~\citep{ComputeAllocation09,Karma}, course registration~\citep{CourseAssignmentHarvard,CourseAssignmentUPenn}, reviewer assignments~\citep{TTT10,CZ13,LSM14}, as well as college admissions and school choice \citep{gale1962college,abdulkadirouglu2005boston,abdulkadirouglu2005new}. 

To reach a desirable allocation, the matching mechanism must account for the preferences of each agent $i$, typically modeled by a valuation vector $(v_{i1}, v_{i2}, \dots, v_{im})$, where $v_{ij}$ represents the value that $i$ has for each item $j$. However, most practical mechanisms are \emph{ordinal}: rather than eliciting each agent's valuation vector, they elicit only their ranking of the items from most-preferred to least-preferred (their ordinal preferences). Since rankings do not capture the intensity of the preferences, ordinal mechanisms generally fail to maximize standard objectives like social welfare (the total value of the agents). The \emph{distortion} of a mechanism quantifies this loss by measuring its ability to approximate such objectives under worst-case instances (see \cite{AFSV21} for a survey of prior work on distortion).

A key limitation of the distortion literature is its reliance on adversarial analysis, which often yields overly pessimistic bounds. For example, the optimal distortion for one-to-one matching with respect to the social welfare is $\Theta(n^2)$ for deterministic mechanisms~\citep{ABFV22} and $\Theta(\sqrt{n})$ for randomized ones~\citep{RFZ14}, even if the valuations are normalized so that $\sum_j v_{ij}=1$. Seeking more practical guarantees, some recent work has turned to \emph{average case analysis} instead: \citet{DGZ22} showed that the Random Priority mechanism achieves a constant \emph{average distortion} for one-to-one matching when each $v_{ij}$ is drawn from a uniform distribution, and \citet{GZ19} extended this result by allowing each $v_{ij}$ to be drawn i.i.d.\ from any given distribution ${F}$. However, in practice, these distributions are often heterogeneous and unknown. We bridge this gap by focusing on the design and analysis of \emph{prior-independent} mechanisms, i.e., mechanisms that are agnostic to these distributions, and evaluating their performance over instances generated by \emph{agent-specific distributions}. In particular, we study the following question:

\begin{center}
\emph{To what extent can ordinal prior-independent mechanisms approximate the optimal social\\welfare when the valuation vectors are drawn from agent-specific distributions?}    
\end{center}

\subsection{Our Results}
Rather than assuming that each $v_{ij}$ is drawn i.i.d.\ from a shared distribution ${F}$, we consider agent-specific distributions. In fact, we let each agent $i$'s valuation vector $(v_{i1}, v_{i2}, \dots, v_{im})$ be drawn simultaneously from $F_i$, allowing the random values in this vector to be highly correlated, and all we require is that each subset of $b_i$ items is equally likely to be $i$'s ``favorite'' bundle. This is clearly satisfied by independent draws and any exchangeable distribution, since they are symmetric with respect to items, but the support of our distributions can be highly non-symmetric. 
Given only the agents' ordinal preferences as input and without any information regarding the distributions that generated them, the mechanism needs to choose an allocation. We aim to design mechanisms with low distortion, i.e., with an expected social welfare that closely approximates the optimal social welfare even if the agent quotas and the agent-specific distributions are chosen adversarially.

\paragraph{\textbf{Optimal distortion bounds (Section~\ref{sec:distortion}).}} We first focus on the special case of one-to-one matching and show that no ordinal mechanism can achieve a distortion better than $e/(e-1)\approx 1.582$. This impossibility holds even if all the values are drawn i.i.d.\ from a \emph{shared} Bernoulli distribution ${F}$ that is \emph{known} to the mechanism. Then, for the general setting, we propose the \emph{Random Survivors} (RS) mechanism. Our main result in this section shows that for any quota vector $(b_1, b_2, \dots, b_n)$ and any distribution profile $\mathcal{F}=(F_1, F_2, \dots, F_n)$, the RS mechanism guarantees a distortion of $e/(e-1)$, matching the impossibility result despite the value heterogeneity and its prior-independence.

The RS mechanism asks each agent $i$ to report only the (unordered) set of their top $b_i$ items (their ``favorite'' items), and it guarantees that every agent will receive each of these items with probability at least $1-1/e$. It, therefore, requires significantly less information than the agents' full rankings, and it ensures that every agent's expected value closely approximates even their favorite bundle value.\footnote{In most cases, it is infeasible for every agent to simultaneously receive all their favorite items, due to conflicts.} 
To achieve this guarantee, the mechanism carefully manages the relative priorities among agents with different quotas. Specifically, it starts by selecting a random subset of the agents, called the ``survivors,'' and giving them higher priority. Each agent $i$ joins the survivors with probability that is a decreasing function of $b_i$ to ensure that agents with lower quotas are not overwhelmed by those with higher quotas. Then, for each item $g$, the mechanism identifies the set of survivors who reported $g$ as one of their favorite items and allocates it to one of them selected uniformly at random.

\paragraph{\textbf{Distortion gap (Section~\ref{sec:dist_gap}).}} Aiming for a more refined analysis, we evaluate the mechanism's performance as a function of the quota vector $\mathbf{b}$. Specifically, we define the \emph{distortion gap} of a mechanism ${\mathcal{M}}$ as the worst-case ratio across all quota vectors $\mathbf{b}$ of the distortion achieved by ${\mathcal{M}}$ on instances induced by $\mathbf{b}$ over the best possible distortion on those instances. We first observe that even though the distortion of the RS mechanism is optimal for the one-to-one matching case ($\mathbf{b} = \mathbf{1}$), its distortion gap is at least $1.5$, i.e., there exists some $\mathbf{b}$ such that a distortion significantly better than $e/(e-1)$ is possible within the instances induced by $\mathbf{b}$. To achieve a smaller distortion gap, we propose the RSBS mechanism, which modifies RS by adding a ``burning'' phase and a ``stealing'' phase. In its first phase, RSBS sets aside the agent with the largest quota and runs the RS mechanism. In the second phase, each agent is ``burnt'' with some probability, causing them to lose any items they received in phase 1. Finally, in phase 3, the agent with the largest quota is provided with the opportunity to ``steal'' all of their favorite items. This mechanism achieves an almost optimal distortion gap of $1.076$ for any distribution profile, thus achieving near-optimal distortion for any given quota profile. In fact, we show that it also maintains the optimal $e/(e-1)$ distortion bound of the RS mechanism, achieving the ``best of both worlds.''

\paragraph{\textbf{Sequential one-pass mechanisms (Section~\ref{sec:sequential}).}} 
We then consider even more demanding settings where the agents are approached sequentially, and their assignment needs to be decided irrevocably before the preferences of subsequent agents are revealed. We show that as long as the agent $i$ with the largest quota is approached last, there exists a mechanism (the HQL mechanism) that simultaneously achieves a distortion $2$ and a distortion gap of $2 (1-1/e)\approx 1.264$. We complement this positive result by showing that no sequential mechanism can achieve a distortion better than $2$ or a distortion gap better than $2 (1-1/e)\approx 1.264$ even if it has full control over the order in which the agents are approached, so the bounds achieved by HQL are optimal. We finally also consider the ``secretary model,'' where the order in which the agents are considered is chosen uniformly at random, and we observe that the RS mechanism from Section~\ref{sec:distortion} can be adapted to this model, maintaining the same distortion of $e/(e-1)\approx 1.582$.\footnote{This is in contrast to the RSBS mechanism, which cannot be implemented in the secretary model.} On the other hand, we show that no ordinal mechanism can achieve a distortion gap better than $4/3$ in the secretary model.

\paragraph{\textbf{Additional observations (Section~\ref{sec:conclusion})}}
Although our main focus is on distortion rather than incentives, we observe that all our mechanisms satisfy Bayesian incentive compatibility (BIC) for any distribution that we consider, thus incentivizing the agents to report truthfully. Furthermore, we observe that all our distortion and distortion gap guarantees apply not only to the class of additive valuations, but also to submodular valuation functions, as long as each set of $b_i$ items is equally likely to be the one that maximizes the randomly generated submodular valuation function.

\subsection{Related Work}
Starting from the work of \citet{PR06} and \citet{BCHLPS15}, a lot of prior work has focused on the distortion measure in a wide variety of different settings (see \cite{AFSV21} for a survey of some of these results). 

The utility maximization variant of one-to-one matching problems has been mostly studied in settings where the agents are assumed to have normalized values over the items (the values of each agent over the set of items sum up to $1$). The best possible distortion for deterministic mechanisms has been identified as $\Theta(n^2)$~\citep{ABFV22} and for randomized mechanisms as $\Theta(\sqrt{n})$~\citep{RFZ14}. The distortion of one-to-one matching and some generalizations has also been studied in settings where the mechanism can elicit more information via queries~ \citep{ABFV21,ABFV22,ABFV24,ma2021matching,LV25,ebadian2025bit}.
An array of publications studied the metric distortion of related matching problems where the agents and items lie in some metric space which induces \emph{metric utilities} for the agents and the goal is to maximize these utilities \citep{Anshelevich2019bipartitematching,Anshelevich2016truthful,anshelevich2016blind,anshelevich2017randomized,ABEPS18}.
Generalizations of matching problems like the ones we consider in this paper have also been studied in the distortion literature, where the mechanism is allowed to also use some value queries~\citep{ABFV21,ABFV22,ABFV24} or request approval information from the agents~\citep{LV25}.

Beyond the scope of matching, more general graph problems \citep{AA18} as well as voting \citep{BNPS21} have been the object of study in the distortion literature.
Moreover, \citet{MPSW19,MSW20} consider the trade-offs between communication and distortion in voting problems.
Another closely related work in voting by \citet{CF24} focused on impartial culture electorates where the preferences of the agents over candidates is generated uniformly at random.

The distortion of matching problems has also received a lot of attention for cost minimization where the costs correspond to distances in some metric space. In \emph{metric matching}, agents and items are embedded in a metric space, and an agent's cost in a matching is their distance from their match. The distortion of one-to-one metric matching was first studied with respect to the utilitarian social cost (the total cost) leading to distortion upper bounds of $n$ for randomized algorithms~\citep{caragiannis2024augmentation} and $O(n^2)$ for deterministic ones~\citep{anari2023matching}, and a distortion lower bound of $\Omega(\log n)$ for both~\citep{anari2023matching}. On the other hand, the optimal distortion of metric matching on the line was shown to be equal to $3$ \citep{FGLRV25}. 
The distortion of metric matching has also recently been analyzed with respect to the egalitarian social cost objective \citep{FGLRV25,hastings2025fairmetricdistortionmatching}.

Recognizing the limitations of worst-case analysis, prior work in the distortion literature focused on the smoothed and average case analysis of the Random Priority mechanism (also known as Random Serial Dictatorship) from \citet{DGZ22} and \citet{GZ19}.
More recently, \citet{LVY26} utilized smooth analysis for the online version of the metric matching problem.
Also, \citet{GKPSZ23} recognized the significance of prior-independent mechanisms and evaluated the expected distortion of such mechanisms in the context of voting.

\section{Preliminaries}

We consider one-to-many matching settings involving $n$ {\em agents} and $m$ {\em items}. We identify agents and items by positive integers and use the sets $[n]=\{1, 2, ..., n\}$ and $[m]=\{1, 2, ..., m\}$ to refer to the sets of agents and items, respectively. Each agent has a positive integer 
\emph{quota} $b_i$, and a {\em matching} is an assignment of items to agents so that each agent $i$ gets up to $b_i$ items and each item is assigned to at most one agent. We assume that $m=\sum_{i=1}^n{b_i}$ so that a matching in which all quotas are satisfied (i.e., each agent $i\in [n]$ is assigned exactly $b_i$ items) is feasible. We denote by $\bb=(b_1, b_2, ..., b_n)$ the quota vector and often use the term $\bb$-matching instance to refer to input instances of our setting. 

Each agent $i\in [n]$ has a non-negative cardinal {\em value} $v_{ij}$ for item $j\in [m]$ and an additive valuation $v_i(S)=\sum_{j\in S} v_{ij}$ for each bundle of items $S\subseteq [m]$.\footnote{In Section~\ref{sec:conclusion}, we discuss how we can relax the additivity assumption and allow for submodular valuations.} We use the $n\times m$ matrix $V=\{v_{ij}\}_{i\in [n], j\in [m]}$ to denote the {\em valuation profile} which consists of the agent values for the items. Matchings are evaluated in terms of their {\em social welfare}. Consider a $\bb$-matching instance and a matching $M$ which assigns the set of items $M_i$ to agent $i\in [n]$. The social welfare of $M$ is simply the total value the agents have for the items they get in $M$, i.e., 
$$\text{SW}(M,V)=\sum_{i\in [n]}{v_i(M_i)} = \sum_{i\in [n]}{\sum_{j\in M_i}{v_{ij}}}.$$ 
We denote by $\text{OPT}(V,\bb)$ the optimal social welfare, i.e., the maximum social welfare over all matchings of the $\bb$-matching instance with valuation profile $V$.

We assume that the valuations of each agent $i$ are drawn from an agent-specific distribution $F_i$, and $\mathcal{F}=(F_1, ..., F_n)$ denotes the agent distribution profile. In its simplest form, the distribution $F_i$ returns a random scalar and agent $i$ draws her valuation $v_{ij}$ for each item $j\in [m]$ independently from $F_i$.
However, our positive results hold even if the whole valuation vector $(v_{i1}, v_{i2},\dots, v_{im})$ of each agent $i$ is drawn simultaneously from $F_i$, allowing for significant correlation across its values, as long as each bundle of $b_i$ items is equally likely to be $i$'s ``favorite'' bundle, i.e., the $b_i$ items that $i$ values the most; we refer to such distributions as {\em unbiased-favorites} (UF). Note that the UF condition is clearly satisfied by any distribution that is \emph{exchangeable}, i.e., symmetric with respect to the items, but the support of UF distributions can be a highly non-symmetric set of valuation vectors.

We consider {\em ordinal matching mechanisms} that do not have access to valuation profiles but instead work with ordinal preference profiles defined as follows. The realized valuations of each agent $i\in [n]$ induce a ranking $\succ_i$ over the items, from most-preferred to least-preferred, i.e., $v_{ij_1}>v_{ij_2}$ implies $j_1\succ_i j_2$. In case agent $i$ has the same value for each item in a set $S$, the tie is resolved randomly, and the ordinal preference of agent $i$ has the items in $S$ ranked in a uniformly random order according to $\succ_i$. We denote by $P(V)$ the random ordinal preference profile that corresponds to the valuation profile $V$. A matching mechanism $\mathcal{M}$ takes as input an ordinal preference profile and returns a matching.

The distortion of a mechanism 
on a $\bb$-matching instance (with respect to the expected welfare) is defined as
\[\text{dist}\left(\mathcal{M},\bb\right) = \sup_{\mathcal{F}}{\frac{\E_{V\sim \mathcal{F}}\left[\text{OPT}(V,\bb)\right]}{\E_{V\sim \mathcal{F}}{\left[\text{SW}(\mathcal{M}(P(V),\bb),V)\right]}}},\]
i.e., as the worst-case ratio over all sets of UF distributions of the expected optimal social welfare over the expected social welfare of the mechanism. The expectations are taken considering the randomness in the selection of the valuation profile and the random tie-breaking in the formation of the ordinal preference profile, as well as possible randomization that is used by mechanism $\mathcal{M}$.
Then, with some abuse in notation, the distortion of mechanism $\mathcal{M}$ is just $$\text{dist}(\mathcal{M})=\sup_{\bb}{\text{dist}(\mathcal{M},\bb)}.$$

By definition, distortion is a pessimistic quantification of efficiency. As we shall see later, the distortion of any matching mechanism is at least $\frac{e}{e-1}$, just because no mechanism can achieve a better distortion on one-to-one matching instances. To obtain more informative results that take the inherent difficulty of different matching instances into account, we introduce the {\em distortion gap}, which compares the distortion of a matching mechanism on each instance to the best distortion that can be achieved on that instance. Specifically, let $\mathcal{M}^*$ be a matching mechanism that achieves the best possible distortion on every matching instance. Then, the distortion gap of the matching mechanism $\mathcal{M}$ is defined as
\begin{align*}
    \text{distgap}(\mathcal{M}) &= \sup_{\bb}{\frac{\text{dist}(\mathcal{M},\bb)}{\text{dist}(\mathcal{M}^*,\bb)}}.
\end{align*}
Clearly, since $\text{dist}(\mathcal{M}^*,\bb)\geq 1$, the distortion of a mechanism is a trivial upper bound on its distortion gap, i.e., $\text{distgap}(\mathcal{M}) \leq \sup_{\bb}{\text{dist}(\mathcal{M},\bb)}=\text{dist}(\mathcal{M})$.

\section{Tight Distortion Bounds}\label{sec:distortion}

In this section, we resolve the distortion of the matching problems that we consider in a strong sense. First, we prove that no ordinal mechanism achieves a distortion better than $e/(e-1)$ on one-to-one matching instances even if every agent's value for each item is drawn independently from the same known univariate distribution. Then, we provide a novel ordinal mechanism that achieves a distortion of $e/(e-1)$ for any quota vector and even if the values of each agent $i$ are drawn from an arbitrary agent-specific and unknown UF distribution $F_i$.

\subsection{A Lower Bound for One-to-One Matching}\label{subsec:lower-bound}
We begin with a lower bound that applies even to one-to-one matching instances. 

\begin{theorem}\label{thm:distortion-lower-bound}
    There exists a univariate distribution $F$ such that even if every agent's values for the items are drawn i.i.d.\ from $F$, no ordinal mechanism\footnote{We remark that all our lower bounds hold for mechanisms that may use randomization.} can achieve a distortion better than $\frac{e}{e-1}$, even if it knows $F$ in advance.
\end{theorem}

\begin{proof}
Consider a one-to-one matching instance with $n$ agents/items. We prove the theorem using an underlying random valuation profile $V$ produced when agents draw their item values i.i.d. from a common univariate distribution $F$ (hence, $\mathcal{F}=\{F, ..., F\}$). In particular, $F$ returns $1$ with probability $1/n^2$ and $0$ otherwise. The proof follows from the next two lemmas. The first one lower-bounds the expected social welfare of the optimal matching.

\begin{lemma}\label{lem:distortion-lower-bound-opt}
    $\E_{V\sim \mathcal{F}}[\text{OPT}(V)] \geq 1-\frac{2}{n}$.
\end{lemma}

\begin{proof}
We say that agent $i$ {\em likes} item $g$ if $i$ ranks $g$ at the top of her ranking and has value $1$ for it. We will bound $\text{OPT}(V)$ by considering matchings in which each item is either assigned to an agent who likes it or stays unmatched if no such agent exists. Then, since each agent likes at most one item, the contribution of each item $g$ to the social welfare of any such matching $M$ is $1$ if there is some agent who likes it and is $0$ otherwise.

    Consider an item $g$ and an agent $i$. By the definition of distribution $F$, we have that the probability that agent $i$ likes item $g$ is $\frac{1}{n}\cdot \left(1-\left(1-\frac{1}{n^2}\right)^n\right)$. This stems from the fact that agent $i$ has value $1$ for at least one item with probability $1-\left(1-\frac{1}{n^2}\right)^n$ and, then, each item can be ranked at the top of agent $i$ equiprobably. Recall that all values are drawn independently. 
    
    We are now ready to bound the contribution of item $g$ to the social welfare of matching $M$. We will need the next technical claim.

\begin{claim}\label{claim:convexity}
    For every $z\in [0,1]$ and $k\geq 1$, it holds $1-(1-z)^k\geq k\cdot z\cdot(1-z)^k$.
\end{claim}    

\begin{proof}
    Since the function $y^k$ is convex, we have that the slope $\frac{1-(1-z)^k}{z}$ of the line connecting points $(1-z,(1-z)^k)$ and $(1,1)$ is lower-bounded by the value of the derivative $k\cdot y^{k-1}$ of $y^k$ at point $(1-z)$. Hence,
    \begin{align*}
        1-(1-z)^k &\geq k\cdot z\cdot(1-z)^{k-1} \geq k\cdot z\cdot (1-z)^k,
    \end{align*}
    and the claim follows.
\end{proof}

We have that the probability $p_g$ that some agent likes item $g$ is
\begin{align*}
    p_g &=1-\left(1-\frac{1}{n}\cdot \left(1-\left(1-\frac{1}{n^2}\right)^n\right)\right)^n \geq \left(1-\left(1-\frac{1}{n^2}\right)^n\right) \cdot \left(1-\frac{1}{n}\cdot \left(1-\left(1-\frac{1}{n^2}\right)^n\right)\right)^n\\
    &\geq\left(1-\left(1-\frac{1}{n^2}\right)^n\right)\cdot \left(1-\frac{1}{n^2}\right)^n \geq \frac{1}{n}\cdot \left(1-\frac{1}{n^2}\right)^{2n}\geq \frac{1}{n}\cdot \left(1-\frac{2}{n}\right).
\end{align*}
The first inequality follows by applying Claim~\ref{claim:convexity} with $k=n$ and $z=\frac{1}{n}\cdot \left(1-\left(1-\frac{1}{n^2}\right)^n\right)$. The second one follows from the property $(1-z)^k\geq 1-k\cdot z$ for $z\in [0,1]$ and $k\geq 1$ (using $k=n$ and $z=\frac{1}{n}\cdot \left(1-\left(1-\frac{1}{n^2}\right)^n\right)$). The third one follows by applying Claim~\ref{claim:convexity} with $k=n$ and $z=\frac{1}{n^2}$. Finally, the fourth inequality follows again from the property $(1-z)^k\geq 1-k\cdot z$ (using $k=2n$ and $z=\frac{1}{n^2}$).
Hence, 
\begin{align*}
    \E_{V\sim \mathcal{F}}[\text{OPT}(V)] \geq \E_{V\sim \mathcal{F}}[M]= \sum_{g\in [n]}{p_g}\geq 1-\frac{2}{n},
\end{align*}
and the lemma follows.
\end{proof}

The next lemma upper-bounds the expected social welfare of the matching returned by any matching mechanism.
\begin{lemma}\label{lem:distortion-lower-bound-sw}
For every one-to-one ordinal matching mechanism $\mathcal{M}$, $\E_{V\sim \mathcal{F}}[\text{SW}(\mathcal{M}(P(V)),V)] \leq 1-\frac{1}{e}+\frac{2}{n}$.
\end{lemma}

\begin{proof}
Consider a mechanism $\mathcal{M}$ and let $a_g$ denote the agent who is assigned item $g$ by $\mathcal{M}$ (if any). If no agent is assigned item $g$, we adopt the convention $v_{a_gg}=0$. Also, we denote by $\text{rank}_i(g)$ the position item $g$ has in the ranking of agent $i$. First, notice that the expected value of the top-ranked item of agent $i$ is equal to the probability that agent $i$ draws at least one $1$ from $F$, i.e.,
\begin{align}\label{eq:top-ranked}
    \E_{V\sim \mathcal{F}}[v_{ig}|\text{rank}_i(g)=1] &=1-\left(1-\frac{1}{n^2}\right)^n\leq \frac{1}{n}.
\end{align}
The inequality follows by applying the property $(1-z)^k \geq 1-k\cdot z$ for $z\in [0,1]$ and $k\geq 1$ (using $k=n$ and $z=\frac{1}{n^2}$).

Furthermore, the expected value of any non-top-ranked item of agent $i$ is upper-bounded by the expected value of its second-ranked item, which in turn is equal to the probability that agent $i$ draws at least two $1$s from $F$, i.e., 
\begin{align}\nonumber
    \E_{V\sim \mathcal{F}}[v_{ig}|\text{rank}_i(g)>1] &\leq \E_{V\sim \mathcal{F}}[v_{ig}|\text{rank}_i(g)=2]\\\nonumber
    &= 1-\left(1-\frac{1}{n^2}\right)^n-n\cdot \frac{1}{n^2}\cdot \left(1-\frac{1}{n^2}\right)^{n-1}\\\label{eq:non-top-ranked}
    &\leq 1-\left(1-\frac{1}{n}\right)-\frac{1}{n}\cdot \left(1-\frac{n-1}{n^2}\right)\leq \frac{1}{n^2}.
\end{align}
Again, the second inequality follows by the same property mentioned above. Hence, we can bound the expected social welfare of the matching returned by mechanism $\mathcal{M}$ as follows:
\begin{align*}
\E_{V\sim \mathcal{F}}[\text{SW}(\mathcal{M}(P(V)),V)] &= \E_{V\sim \mathcal{F}}\left[\sum_{g\in [n]}{v_{a_gg}}\right]= \sum_{g\in [n]}{\E_{V\sim \mathcal{F}}[v_{a_gg}]}\\
&=\sum_{g\in [n]}{\left(\E_{V\sim \mathcal{F}}[v_{a_gg}|\text{rank}_{a_g}(g)=1]\cdot \Pr[\text{rank}_{a_g}(g)=1]\right.}\\
&\quad\quad\quad\quad {\left.+\E_{V\sim \mathcal{F}}[v_{a_gg}|\text{rank}_{a_g}(g)>1]\cdot \Pr[\text{rank}_{a_g}(g)>1]\right)}\\
&\leq \sum_{g\in [n]}{\left(\frac{1}{n}\cdot \Pr[\text{rank}_{a_g}(g)=1]+\frac{1}{n^2}\right)}.
\end{align*}
Observe that $\Pr[\text{rank}_{a_g}(g)=1]$ is upper-bounded by $1-\left(1-\frac{1}{n}\right)^n$, which is the probability that item $g$ is ranked first by some agent. Hence,
\begin{align*}
\E_{V\sim \mathcal{F}}[\text{SW}(\mathcal{M}(P(V)),V)] &\leq 1-\left(1-\frac{1}{n}\right)^n+\frac{1}{n} = 1-\left(1-\frac{1}{n}\right)^{n-1}+\frac{1}{n}\cdot \left(1-\frac{1}{n}\right)^{n-1}+\frac{1}{n}\\
&\leq 1-\frac{1}{e}+\frac{2}{n},
\end{align*}
completing the proof. The last inequality follows since $(1-1/n)^{n-1}\geq 1/e$ for every $n\geq 1$.
\end{proof}
Now, the definition of distortion and Lemmas~\ref{lem:distortion-lower-bound-opt} and~\ref{lem:distortion-lower-bound-sw} yield
\begin{align*}
\text{dist}(\mathcal{M}) &\geq \frac{\E_{V\sim \mathcal{F}}[\text{OPT}(V)]}{\E_{V\sim \mathcal{F}}[\text{SW}(\mathcal{M}(P(V)),V)]} \geq \frac{1-\frac{2}{n}}{1-\frac{1}{e}+\frac{2}{n}}.
\end{align*}
The theorem follows since the RHS of the above inequality approaches $\frac{e}{e-1}$ as $n$ goes to infinity.
\end{proof}

\subsection{The Random Survivors Mechanism}\label{subsec:rs}

We now present our mechanism for $\bb$-matching instances, which we call {\em Random Survivors} (RS). 
Mechanism RS takes as input a set of items $m$ and a set of $n$ agents with integer quota $b_i\geq 1$ for agent $i\in [n]$ such that $\sum_{i\in [n]}{b_i}=m$. For each $i\in [n]$, RS declares agent $i$ as a survivor 
with probability $p_i:=1-\frac{b_i-1}{3m}$, independently of all other agents. The survivors will be the only agents who can get items. 
For each item $g\in [m]$, let $D_g$ be the set of survivors $i$ that have item $g$ among their top $b_i$ items. If $D_g$ is non-empty, mechanism RS matches item $g$ to one of the agents in $D_g$, selected uniformly at random. Mechanism RS is depicted as Algorithm~\ref{alg:rs} below. From now on, we refer to the top $b_i$ items of agent $i$ as their favorite items.\footnote{We describe our mechanisms as assigning up to $b_i$ items to each agent $i$ and, consequently, leaving some of the items unassigned. Of course, when implementing the mechanisms in practice, an additional final phase can be used to distribute the unassigned items to agents so that each agent covers their quota exactly. }  

\begin{algorithm}[h]
    \caption{The Random Survivors Mechanism}
    \label{alg:rs}
    \begin{algorithmic}[1]
        \Input A \textbf{b}-matching instance with \(m\) items and \(n\) agents. Query access to the set $\tp_i$ containing the favorite items of agent $i\in [n]$
        \Output A \textbf{b}-matching
        \State \(S\gets \emptyset\)
        \ForEach{\(i \in [n]\)}
            \State With probability \(1 - (b_i - 1) / 3m\): \(S\gets S\cup \{i\}\)
        \EndFor
        \State \(M\gets \emptyset\)
        \ForEach{\(g \in [m]\)}
            \State \(D_g \gets \{i\in S: g\in \tp_i\}\)
            \If{\(D_g \neq \emptyset\)}
                \State Choose \(i^* \in D_g\) uniformly at random
                \State \(M \gets M \cup \{(i^*, g)\}\)
            \EndIf
        \EndFor
        \State\Return \(M\)
    \end{algorithmic}
\end{algorithm}

\paragraph{Warming-up (one-to-one matching instances).}
As a warm-up, we first evaluate the distortion of RS for the special case of one-to-one matching. In this case, all agents are survivors, and the mechanism allocates each item equiprobably among the agents that reported it as their favorite. 

\begin{theorem}\label{thm:rs-warm-up}
    The distortion of mechanism RS on one-to-one matching instances is at most $\frac{e}{e-1}$.
\end{theorem}

\begin{proof}
Consider the application of mechanism RS to a one-to-one matching instance with $n$ agents and items and random valuation profile $V$ drawn from the profile of UF distributions $\mathcal{F}=(F_1, ..., F_n)$.

Consider an agent $i\in [n]$ and let item $g$ be the top choice in their ranking. For $j\in [n]\setminus\{i\}$ denote by $X_j$ the indicator random variable denoting whether agent $j$'s top choice is also item $g$ ($X_j=1$) or not ($X_j=0$). Clearly, given the values of the random variables, the probability that mechanism RS assigns item $g$ to agent $i$ is $\frac{1}{1+\sum_{j\in [n]\setminus \{i\}}{X_j}}$. Hence, the probability that mechanism RS assigns agent $i$ their top-choice item is
    \begin{align*}
        \E\left[\frac{1}{1+\sum_{j\in [n]\setminus\{i\}}{X_j}}\right] &= \E\left[\int_0^1{z^{\sum_{j\in [n]\setminus\{i\}}{X_j}}\ud{z}}\right]=\int_0^1{\E\left[z^{\sum_{j\in [n]\setminus\{i\}}{X_j}}\right]}\\
        &= \int_0^1{\prod_{j\in [n]\setminus \{i\}}{\E\left[z^{X_j}\right]\ud{z}}}=\int_0^1{\left(1-\frac{1}{n}+\frac{z}{n}\right)^{n-1}\ud{z}}\\
        &= 1 - \left(1-\frac{1}{n}\right)^{n} \geq 1-1/e.
    \end{align*}
The first equality follows since $\int_0^1{z^d\ud{z}}=\frac{1}{d+1}$ for $d\geq 0$, the second one by linearity of expectation, and the third one since the random variables $X_j$ are independent for different $j\in [n]\setminus \{i\}$. The fourth inequality follows since $X_j=1$ with probability $1/n$.
    
Now, denote by $v_i$ the random variable indicating the valuation of agent $i$ for their top item. Then, by just considering the value the agents get from the assignment of their top item, the expected social welfare of the matching computed by mechanism RS is
    \begin{align*}
        \E_{V\sim \mathcal{F}}[SW(\mathcal{M}_{RS}(P(V)))] &\geq \E\left[\sum_{i\in [n]}{\left(1-1/e\right)\cdot v_i}\right]=\left(1-1/e\right)\cdot \E\left[\sum_{i\in [n]}{v_i}\right]\\
        &\geq \left(1-1/e\right)\cdot \E_{V\sim \mathcal{F}}[OPT(V)].
    \end{align*}
The lemma follows from the definition of the distortion. The equality in this last derivation follows by linearity of expectation, and the second inequality follows since no matching can extract more value than $v_i$ from agent $i\in [n]$, and thus $OPT(V) \leq \sum_{i\in [n]}{v_i}$.
    \end{proof}

\paragraph{Analysis for general $\bb$-matching instances.}
We now consider the general case of different quotas per agent. By extending the arguments in the proof of Theorem~\ref{thm:rs-warm-up}, we will prove that mechanism RS achieves the best-possible distortion of $e/(e-1)$ in the general case as well.

\begin{theorem}\label{thm:alg-2}
    The distortion of algorithm RS on any $\bb$-matching instance is at most $\frac{e}{e-1}$.
\end{theorem}
The proof of Theorem~\ref{thm:alg-2}, as well as the proofs of all distortion upper bounds we present in the following, are based on the following key lemma. It asserts that upper bounds on the distortion of mechanisms can be obtained by bounding the probability that each agent gets any of their favorite items.

\begin{lemma}\label{lem:template-upper}
Consider a matching mechanism $\mathcal{M}$ applied on a $\bb$-matching instance with underlying random valuation profile $V$ drawn from the profile of UF distributions $\mathcal{F}$. For $i\in [n]$ and $t\in [b_i]$, let $q_{it}$ be the probability that $\mathcal{M}$ matches agent $i$ with their $t$-th ranked item and let $\rho=\min_{i\in [n], t\in [b_i]}{q_{it}}$. Then, algorithm $\mathcal{M}$ has distortion at most $1/\rho$.
\end{lemma}

\begin{proof}
Consider a $\bb$-matching instance with $m$ items and $n$ agents with integer quota $b_i\geq 1$ for agent $i\in [n]$ such that $\sum_{j\in [n]}{b_j}=m$ and an underlying random valuation profile $V$ drawn from the profile of UF distributions $\mathcal{F}$. For each agent $i\in [n]$ and $t\in [b_i]$ let $v_i^{(t)}$ be the random variable indicating the valuation agent $i$ has for the $t$-th item in their ranking. Then, considering the value each agent gets from their favorite items that mechanism $\mathcal{M}$ assigns to them, the expected social welfare of the matching computed by algorithm $\mathcal{M}$ is
\begin{align*}
    \E_{V\sim \mathcal{F}}\left[SW(\mathcal{M}(P(V),\bb),V)\right] &\geq \E\left[\sum_{i\in [n]}{\sum_{t\in [b_i]}{q_{it}\cdot v_{i}^{(t)}}}\right] \geq  \rho \cdot \E\left[\sum_{i\in [n]}{\sum_{t\in [b_i]}{v_{i}^{(t)}}}\right]\\
    &\geq \rho \cdot \E_{V\sim \mathcal{F}}[OPT(V,\bb)].
\end{align*}
The lemma follows from the definition of distortion. The second inequality in the above derivation follows from the definition of $\rho$ and the third one since no $\bb$-matching
can extract more value than $\sum_{t\in [b_i]}{v_i^{(t)}}$ from agent $i$.
\end{proof}

\begin{proof}[Proof of Theorem~\ref{thm:alg-2}]
Consider a $\bb$-matching instance with $m$ items, and $n$ agents with integer quota $b_i\geq 1$ for agent $i\in [n]$ such that $\sum_{i\in [n]}{b_i}=m$ and a random valuation profile $V$ drawn from the set of UF distributions $\mathcal{F}$. We will show that for every agent $i\in [n]$ and $t\in [b_i]$ the probability that algorithm RS assigns the $t$-th top item of agent $i$ to them is at least $1-1/e$. Then, the theorem will follow from Lemma~\ref{lem:template-upper}. 

Let $g$ be the random variable denoting the item which is the $t$-th top item of agent $i$. For $j\in [n]\setminus \{i\}$, we denote by $X_j$ the binary random variable indicating whether the agent $j$ is a survivor and ranks item $g$ in her top $b_j$ positions of their ranking ($X_j=1$) or not ($X_j=0$). Then, the probability that mechanism RS matches item $g$ to agent $i$ is 
\begin{align}\label{eq:alg-2-q_it}
    q_{it} &=p_i\cdot \E\left[\frac{1}{1+\sum_{j\in [n]\setminus \{i\}}{X_j}}\right]. 
\end{align}
Similarly to the proof of Theorem~\ref{thm:rs-warm-up}, we compute the expectation in the RHS of Equation (\ref{eq:alg-2-q_it}) as follows.
\begin{align}\nonumber
    \E\left[\frac{1}{1+\sum_{j\in [n]\setminus \{i\}}{X_j}}\right] &=\E\left[\int_0^1{z^{\sum_{j\in [n]\setminus \{i\}}{X_j}}\ud{z}}\right]=\int_0^1{\E\left[z^{\sum_{j\in [n]\setminus \{i\}}{X_j}}\right]\ud{z}}\\\nonumber
    &=\int_0^1{\prod_{j\in [n]\setminus \{i\}}{\E\left[z^{X_j}\right]}\ud{z}}=\int_0^1{\prod_{j\in [n]\setminus \{i\}}{\left(1-\frac{b_j\cdot p_j}{m}+\frac{b_j\cdot p_j}{m}\cdot z\right)}\ud{z}}\\\label{eq:alg-2-q_it-exp}
    &= \int_0^1{\prod_{j\in [n]\setminus \{i\}}{\left(1-\frac{b_j\cdot p_j}{m}\cdot y\right)}\ud{y}}.
\end{align}
The first equality follows by the fact $\int_0^1{z^d\ud{z}}=\frac{1}{d+1}$ for $d\geq 0$. The second equality follows by linearity of expectation, and the third one by the fact that the random variables $X_j$ for $j\in [n]\setminus \{i\}$ are independent. The fourth equality follows since the random variable $X_j$ is equal to $1$ with probability $\frac{p_j\cdot b_j}{m}$ and is equal to $0$ otherwise. The last equality follows using the substitution $y=1-z$.

We bound the RHS of Equation (\ref{eq:alg-2-q_it-exp}) using the next lemma, which is stated in a slightly more general form that will be useful later in Section~\ref{sec:dist_gap}.
    \begin{lemma}\label{lem:alg-2-exp-bound}
    For every set of agents $S\subseteq [n]$, it holds that
    \begin{align*}
        \int_0^1{\prod_{j\in S}{\left(1-\frac{b_j\cdot p_j}{m}\cdot y\right)\ud{y}}} &\geq \int_0^1{\left(1-\frac{y}{m}\right)^{\sum_{j\in S}{b_j}}\ud{y}}.
    \end{align*}
    \end{lemma}

    \begin{proof}
    We rename all agents so that $S=\{1, 2, ..., |S|\}$. We will show that for every $t\in S$ it holds that
    \begin{align}\label{eq:int-inequality-to-prove}
        \int_0^1{\prod_{j=1}^t{\left(1-\frac{b_j\cdot p_j}{m}\cdot y\right)}\cdot \left(1-\frac{y}{m}\right)^{\sum_{j=t+1}^{|S|}{b_j}}\ud{y}} &\geq \int_0^1{\prod_{j=1}^{t-1}{\left(1-\frac{b_j\cdot p_j}{m}\cdot y\right)}\cdot \left(1-\frac{y}{m}\right)^{\sum_{j=t}^{|S|}{b_j}}\ud{y}}.
    \end{align}
    The lemma will then follow by combining these inequalities for $t=|S|, |S|-1, ..., 1$.
    
    Notice that inequality (\ref{eq:int-inequality-to-prove}) holds trivially if $b_t=1$ since $b_t\cdot p_t=1$ (and, hence, $1-\frac{b_t\cdot p_t}{m}\cdot y=1-\frac{y}{m}$) in this case. So, in the following, we assume that $b_t\geq 2$. We will use two technical lemmas. For any agent $t\in [n]$, let $f_t$ be the function defined as 
        \begin{align*}
           f_t(y) &= 1-\frac{b_t\cdot p_t}{m}\cdot y-\left(1-\frac{y}{m}\right)^{b_t}.
        \end{align*}
        
        \begin{lemma}\label{lem:int-inequality-to-prove-technical-lemma}
        For every $t\in [n]$ such that $b_t\geq 2$, there exists $y_t\in [0,1]$ such that the function $f_t$ is non-negative in $[0,y_t]$ and non-positive in $(y_t,1]$.
        \end{lemma}

        \begin{proof}
Using the definition $p_t=1-\frac{b_t-1}{3m}$, we have that the derivative of $f_t$ is
        \begin{align*}
            f'_t(y) &= -\frac{b_t\cdot p_t}{m}+\frac{b_t}{m}\cdot \left(1-\frac{y}{m}\right)^{b_t-1}=-\frac{b_t}{m}\cdot \left(1-\left(1-\frac{y}{m}\right)^{b_t-1}\right)+\frac{b_t\cdot (b_t-1)}{3m^2},
        \end{align*}
        i.e.\ strictly decreasing for $y\in [0,1]$ since $b_t\geq 2$. Hence, function $f_t$ is strictly concave and has at most two roots, one of which is at $0$. Also, notice that $f'_t(0)=\frac{b_t\cdot (b_t-1)}{3m^2}>0$. Hence, there are two possibilities for $f_t$. The first one is that $f_t$ has no second root and $f_t(y)> 0$ for $y>0$. Then, the lemma follows by setting $y_t=1$. The second one is that $f_t$ has a second root $r>0$ and satisfies $f_t(y)>0$ for $y\in (1,r)$, $f_t(r)=0$, and $f_t(y)<0$ for $y\in (r,+\infty)$. The lemma then follows by setting $y_t=\min\{1,r\}$.
        \end{proof}

        \begin{lemma}\label{lem:int-inequality-f}
        For every agent $t\in [n]$, it holds that $\int_0^1{f_t(y)\ud{y}} \geq 0$.
        \end{lemma}
    
    \begin{proof}
    The inequality holds trivially as equality when $b_t=1$. To handle the case $b_t\geq 2$, we will use the next technical claim.

    \begin{claim}\label{claim:poly-3}
    Let $k\geq 3$ and $\ell\geq k-1$ be integers. Then,
    \[
        \left( 1 - \frac{1}{\ell} \right)^{k} \geq 1 - \frac{k}{\ell} + \binom{k}{2}\cdot \frac{1}{\ell^2} - \binom{k}{3}\cdot \frac{1}{\ell^3}\, .
    \]
    \end{claim}

    \begin{proof}
    Recall that 
    \begin{align*}
        \left(1-\frac{1}{\ell}\right)^k &= \sum_{j=0}^k{\left(-1\right)^j\cdot \binom{k}{j}\cdot \frac{1}{\ell^j}}.    
    \end{align*}        
    Hence, 
    \begin{align*}
        &\left(1-\frac{1}{\ell}\right)^k -1 + \frac{k}{\ell}- \binom{k}{2}\cdot \frac{1}{\ell^2} + \binom{k}{3}\cdot \frac{1}{\ell^3} = \sum_{j=4}^k{\left(-1\right)^j\cdot \binom{k}{j}\cdot \frac{1}{\ell^j}}\\
        &= \sum_{j=2}^{\left\lfloor\frac{k-1}{2}\right\rfloor}{\binom{k}{2j}\cdot \frac{1}{\ell^{2j}}\cdot \left(1-\frac{k-2j}{(2j+1)\cdot \ell}\right)}+\sum_{j=\left\lfloor \frac{k+1}{2}\right\rfloor}^k{(-1)^j\cdot \binom{k}{j}\cdot \frac{1}{\ell^j}}\\
        &\geq 0,
    \end{align*} 
    as desired. To see why the last inequality is true, observe that both sums in the RHS expression contain no terms if $k=3$. The second sum contains a single term if $k$ is even and at least $4$; this term is clearly non-negative. The first sum contains terms for $k\geq 5$. In this case, since $j\geq 2$, we have $k-2j\leq k-1 \leq (2j+1)\cdot \ell$, which implies that each term in the first sum is non-negative as well.
    \end{proof}

    We now have
    \begin{align*}
        \int_0^1{\left( 1 - \frac{y}{m} \right)^{b_t} \ud{y}} &= \frac{m}{b_t+1}\cdot \left(1-\left(1-\frac{1}{m}\right)^{b_t+1}\right)\\
        &\leq \frac{m}{b_t+1}\cdot \left(\frac{b_t+1}{m}-\binom{b_t+1}{2}\cdot \frac{1}{m^2}+\binom{b_t+1}{3}\cdot \frac{1}{m^3}\right)\\
        &=1-\frac{b_t}{2m}+\frac{b_t\cdot (b_t-1)}{6m^2} = 1-\frac{b_t\cdot p_t}{2m} = \int_0^1{\left(1-\frac{b_t\cdot p_t}{m}\cdot y\right)\ud{y}},
    \end{align*}
    which implies that $\int_0^1{f_t(y)\ud{y}} \geq 0$, completing the proof of Lemma~\ref{lem:int-inequality-f}.
    The inequality follows by applying Claim~\ref{claim:poly-3} with $\ell=m$ and $k=b_t+1\geq 3$ (since $b_t\leq m$, it is $\ell\geq k-1$ as well). The second last equality follows by the definition of $p_t$.
    \end{proof}

        Let
        \begin{align*}
            g_t &=\prod_{j=1}^{t-1}{\left(1-\frac{b_j\cdot p_j}{m}\cdot y\right)}\cdot \left(1-\frac{y}{m}\right)^{\sum_{j=t+1}^{|S|}{b_j}}.
        \end{align*}
        Using the definitions of functions $f_t$ and $g_t$, we have
        \begin{align*}
            &\int_0^1{\prod_{j=1}^t{\left(1-\frac{b_j\cdot p_j}{m}\cdot y\right)}\cdot \left(1-\frac{y}{m}\right)^{\sum_{j=t+1}^{|S|}{b_j}}\ud{y}} -\int_0^1{\prod_{j=1}^{t-1}{\left(1-\frac{b_j\cdot p_j}{m}\cdot y\right)}\cdot \left(1-\frac{y}{m}\right)^{\sum_{j=t}^{|S|}{b_j}}\ud{y}}\\
            &= \int_0^1{\left(1-\frac{b_t\cdot p_t}{m}\cdot y-\left(1-\frac{y}{m}\right)^{b_t}\right)\cdot \prod_{j=1}^{t-1}{\left(1-\frac{b_j\cdot p_j}{m}\cdot y\right)\cdot \left(1-\frac{y}{m}\right)^{\sum_{j=t+1}^{|S|}{b_j}}}\ud{y}}\\
            &=\int_0^1{f_t(y)\cdot g_t(y)\ud{y}} = \int_0^{y_t}{f_t(y)\cdot g_t(y)\ud{y}}+\int_{y_t}^1{f_t(y)\cdot g_t(y)\ud{y}}\\
            &\geq \int_0^{y_t}{f_t(y)\cdot g_t(y_t)\ud{y}}+\int_{y_t}^1{f_t(y)\cdot g_t(y_t)\ud{y}}=g_t(y_t)\cdot \int_0^1{f_t(y)\ud{y}} \geq 0,
        \end{align*}
        which implies inequality (\ref{eq:int-inequality-to-prove}) and completes the proof of Lemma~\ref{lem:alg-2-exp-bound}. The first inequality follows since $g_t$ is non-increasing in $y$ and, by Lemma~\ref{lem:int-inequality-to-prove-technical-lemma}, $f_t(y)$ is non-negative in $[0,y_t]$ and non-positive in $(y_t,1]$ (if non-empty). The last inequality follows by Lemma~\ref{lem:int-inequality-f} and since $g_t(y)$ is non-negative for $y\in [0,1]$.
    \end{proof}

We are now ready to bound the probability $q_{it}$ as follows. We use Equations (\ref{eq:alg-2-q_it}) and (\ref{eq:alg-2-q_it-exp}), Lemma~\ref{lem:alg-2-exp-bound} with $S=[n]\setminus \{i\}$, the fact $\sum_{j\in [n]\setminus \{i\}}{b_j}=m-b_i$, and the definition $p_i=1-\frac{b_i-1}{3m}$ to get
\begin{align}\nonumber
    q_{it} &= p_i \cdot \E\left[\frac{1}{1+\sum_{j\in [n]\setminus\{i\}}{X_j}}\right] = p_i \cdot \int_0^1{\prod_{j\in [n]\setminus \{i\}}{\left(1-\frac{b_j\cdot p_j}{m}\cdot y\right)}\ud{y}}\\\nonumber
    &\geq p_i\cdot \int_0^1{\left(1-\frac{y}{m}\right)^{\sum_{j\in [n]\setminus\{i\}}{b_j}}\ud{y}}= p_i\cdot \int_0^1{\left(1-\frac{y}{m}\right)^{m-b_i}\ud{y}}\\\nonumber
    &=\left(1-\frac{b_i-1}{3m}\right)\cdot \frac{m}{m-b_i+1}\cdot \left(1-\left(1-\frac{1}{m}\right)^{m-b_i+1}\right)\\\label{eq:alg-2-q_it-final}
    &\geq \left(1-\frac{b_i-1}{3m}\right)\cdot \frac{m}{m-b_i+1}\cdot \left(1-\exp\left(-1+\frac{b_i-1}{m}\right)\right).
\end{align}
The second inequality follows since $(1-1/m)^m\leq e^{-1}$. We will bound the RHS of Equation (\ref{eq:alg-2-q_it-final}) using a last technical lemma. Again, this is stated in a slightly more general form than is necessary here; this will also be useful later in Section~\ref{sec:dist_gap}.
\begin{lemma}\label{lem:technical-bound-for-q_it-final}
    For every $z\in [0,1)$, the function 
    \[g_z(x)=\frac{(1-x/3)\cdot (1-\exp(-1+x+z))}{1-x-z}\] is non-decreasing for $x\in [0,1-z)$.
\end{lemma}

\begin{proof}
    The derivative of function $g_z$ with respect to $x$ is
    \begin{align*}
        g'_z(x) &= \frac{2+z+e^{-1+x+z}\cdot (-5+4x+2z-xz-x^2)}{3(1-x-z)^2}.
    \end{align*}
    The derivative of the numerator with respect to $x$ is equal to
    \begin{align*}
        -e^{-1+x+z}\cdot (1-x)\cdot (1-x-z),
    \end{align*}
    i.e., non-positive for $z\in [0,1)$ and $x\in [0,1-z)$. Hence, the numerator of the equivalent expression for $g'_z(x)$ is non-increasing in $x$. Observe that, for $x=1-z$, the numerator evaluates to $0$, and therefore, $g'_z(x)\geq 0$ for $x\in [0,1-z)$.
\end{proof}
Notice that the RHS of Equation (\ref{eq:alg-2-q_it-final}) is equal to $g_0\left(\frac{b_i-1}{m}\right)$. Since $g_0(x)$ is non-decreasing in $x$, we obtain that $q_{it}\geq g_0(0)=1-1/e$. This completes the proof of Theorem~\ref{thm:alg-2}.
\end{proof}

\section{Bounding the Distortion Gap}\label{sec:dist_gap}
We devote this section to developing our approach for upper-bounding the distortion gap of mechanisms. Our main tool is a lower bound on the distortion of any mechanism (Theorem~\ref{thm:benchmark}), which is used as a benchmark in our analysis. The benchmark is expressed in terms of the parameters of a $\bb$-matching instance and serves as a proxy for the distortion of the optimal mechanism $\mathcal{M}^*$ on this instance. Then, upper bounds on the distortion gap are obtained by combining our benchmark with instance-specific bounds on the distortion (see Corollary~\ref{cor:distortion-gap}).

Note that the lower bound on the distortion presented earlier in Section~\ref{sec:distortion} (Theorem~\ref{thm:distortion-lower-bound}) holds specifically for one-to-one matching instances. To extend it for more general $\bb$-matching instances, the proof of Theorem~\ref{thm:benchmark} (and of the intermediate Lemma~\ref{lem:for-proving-lower-bounds}) below exploits UF distributions that have more structure compared to those used in the proof of Theorem~\ref{thm:distortion-lower-bound}.

\begin{theorem}\label{thm:benchmark}
Consider an ordinal matching mechanism $\mathcal{M}$ applied on $\bb$-matching instances with $m$ items and $n$ agents with integer quota $b_i\geq 1$ for agent $i\in [n]$ such that $\sum_{j\in [n]}{b_j}=m$. Then, the distortion of $\mathcal{M}$ satisfies $\text{dist}(\mathcal{M},\bb)\geq \left(1-\prod_{i\in [n]}{\left(1-\frac{b_i}{m}\right)}\right)^{-1}$.
\end{theorem}

Theorem~\ref{thm:benchmark} follows as a corollary of the next lemma, which is also used later in Section~\ref{sec:sequential} to prove lower bounds on the distortion of specific classes of mechanisms.

\begin{lemma}\label{lem:for-proving-lower-bounds}
    Consider an ordinal matching mechanism $\mathcal{M}$ applied on $\bb$-matching instances with $m$ items and $n$ agents with integer quota $b_i\geq 1$ for agent $i\in [n]$ such that $\sum_{j\in [n]}{b_j}=m$. If there exists $i^*\in [n]$ such that the expected number of favorite items mechanism $\mathcal{M}$ assigns to agent $i^*$ is at most $\rho\cdot b_{i^*}$, then the distortion of $\mathcal{M}$ is at least $1/\rho$. 
\end{lemma}

\begin{proof}
    Consider the profile $\mathcal{F}=\{F_1, ..., F_n\}$ of UF distributions which produce a random valuation profile $V$ as follows. For every $i\in [n]\setminus \{i^*\}$, distribution $F_i$ sets the value of agent $i$ for any item equal to $0$. Distribution $F_{i^*}$ selects $b_{i^*}$ of the $m$ items uniformly at random (with replacement), sets the value of agent $i^*$ for each of them equal to $1$, and sets the value of agent $i^*$ for any other item equal to $0$.

    Clearly, we have $\E_{V\sim \mathcal{F}}[OPT(V)]=b_{i^*}$, since assigning to agent $i^*$ their $b_{i^*}$ favorite items yields social welfare $b_{i^*}$. As agent $i^*$ is the only one who can contribute to the social welfare of mechanism $\mathcal{M}$, we have $\E_{V\sim \mathcal{F}}[\text{SW}(\mathcal{M}(P(V),\bb),V)]\leq \rho\cdot b_{i^*}$. By the definition of distortion, we conclude that $\text{dist}(\mathcal{M},\bb)\geq 1/\rho$.
\end{proof}

We are now ready to prove Theorem~\ref{thm:benchmark}.
\begin{proof}[Proof of Theorem~\ref{thm:benchmark}]
Let $\rho=1-\prod_{i\in [n]}{\left(1-\frac{b_i}{m}\right)}$. For the sake of contradiction, assume that mechanism $\mathcal{M}$ has distortion less than $1/\rho$. Then, by Lemma~\ref{lem:for-proving-lower-bounds}, it must be that, for every $i\in [n]$, the expected number of favorite items mechanism $\mathcal{M}$ assigns to agent $i$ is higher than $\rho\cdot b_i$. Hence, by linearity of expectation, the expected total number of favorite items the agents get is higher than $\sum_{i\in [n]}{\rho\cdot b_i}=\rho\cdot m$.

However, consider an item $g\in [m]$. Since the distribution $F_i$ is UF, the probability that item $g$ is the favorite item of agent $i\in [n]$ is $b_i/m$. Hence, the probability that item $g$ is the favorite item of some agent is only $1-\prod_{i\in [n]}{\left(1-\frac{b_i}{m}\right)}=\rho$. Hence, the expected number of items that are the favorite item of some agent is exactly $\rho\cdot m$, contradicting the conclusion of the previous paragraph.
\end{proof}

Using Theorem~\ref{thm:benchmark}, we can recover the $\left(1-\left(1-1/n\right)^n\right)^{-1}\approx \frac{e}{e-1}$ lower bound of Theorem~\ref{thm:distortion-lower-bound} for one-to-one matchings. For general $\bb$-matching instances, the worst-case distortion $\frac{e}{e-1}$ we proved for algorithm RS can be far from the per-instance distortion lower bound indicated by Theorem~\ref{thm:benchmark}. As an example, consider the extreme case of an instance with a large number of items $m$ and two agents with quotas $m-1$ and $1$, respectively. We will refine our results by using the distortion lower bound from Theorem~\ref{thm:benchmark} as a benchmark, against which the per-instance distortion of mechanisms is compared. Specifically, Theorem~\ref{thm:benchmark} yields the following statement, which we will use extensively from now on.

\begin{corollary}\label{cor:distortion-gap}
The distortion gap of any ordinal matching mechanism $\mathcal{M}$ satisfies    
\[\text{distgap}(\mathcal{M}) \leq \sup_{\bb}{\text{dist}(\mathcal{M},\bb)\cdot \left(1-\prod_{i=1}^n{\left(1-\frac{b_i}{m}\right)}\right)}.\]
\end{corollary}
Our aim from now on is to design mechanisms with (almost) optimal distortion and a distortion gap as close to $1$ as possible. Unfortunately, mechanism RS falls short regarding the second goal.\footnote{To see why, consider again the extreme instance above with $m$ items and two agents with quotas $m-1$ and $1$. Assuming large $m$, we can prove that almost optimal distortion is possible, e.g., by assigning her favorite item to agent 2 and the favorite items that are still available to agent 1. However, according to the definition of mechanism RS, agent 1 is a survivor with probability very close to $2/3$ and, by Lemma~\ref{lem:for-proving-lower-bounds}, the distortion and, consequently, distortion gap are arbitrarily close to $3/2$.} Our next mechanism RSBS, standing for {\em random survivors with item burning and stealing}, blends the ideas of mechanism RS with new ones and achieves both guarantees.

RSBS takes as input a $\bb$-matching instance. Let $b_{\max}=\max_{i\in [n]}{b_i}$ and let $i^*$ be an agent with quota $b_{\max}$. RSBS runs in three phases: 
\begin{itemize}
    \item In phase 1, it runs mechanism RS, ignoring agent $i^*$. I.e., each agent $i\in [n]\setminus \{i^*\}$ is declared as a survivor independently with probability $p_i:=1-\frac{b_i-1}{3m}$. Each item $g$ is assigned to an agent selected uniformly at random from the set $D_g$, which contains each survivor $i$ who has item $g$ as favorite. Of course, item $g$ is not assigned to any agent if the set $D_g$ is empty. 
    
    \item In phase 2, with probability
\begin{align}\label{eq:burning-prob}
    \beta_i &:= 1-\frac{1-\exp\left(-1+\frac{b_{\max}}{m}\right)}{\left(1-\frac{b_{\max}}{m}\right)\cdot p_i \cdot \int_0^1{\prod_{j\in [n]\setminus \{i,i^*\}}{\left(1-\frac{p_j\cdot b_j}{m}\cdot y\right)}\ud{y}}},
\end{align}
and independently for each survivor $i\in [n]\setminus \{i^*\}$, RSBS cancels (or ``burns'') $i$'s assignment from phase 1 and these items become unassigned again.

\item In phase 3, agent $i^*$ is first assigned any of her favorite items that are currently unassigned. Then, agent $i^*$ ``steals'' all items currently assigned to some agent in $[n]\setminus \{i^*\}$ with probability 
\begin{align}\label{eq:stealing-prob}
    \sigma &:= \frac{1-\left(2-\frac{b_{\max}}{m}\right)\cdot \exp\left(-1+\frac{b_{\max}}{m}\right)}{1-\exp\left(-1+\frac{b_{\max}}{m}\right)}.
\end{align}
\end{itemize}

Before proceeding with the analysis of RSBS, we show that it is well-defined, i.e., the quantities $\beta_i$ and $\sigma$ as defined in Equations (\ref{eq:burning-prob}) and~(\ref{eq:stealing-prob}) above are always between $0$ and $1$ and can indeed be used as probabilities. In the proof of the next statement, we exploit Lemmas~\ref{lem:alg-2-exp-bound} and~\ref{lem:technical-bound-for-q_it-final} from Section~\ref{subsec:rs}.

\begin{lemma}\label{lem:rsbs-well-defined}
    Mechanism RSBS is well-defined.
\end{lemma}

\begin{proof}
We first show that $\beta_i\in [0,1)$ for $i\in [n]\setminus\{i^*\}$. Since $b_i<m$ for $i\in [n]$ and $p_i\in (0,1]$ for $i\in [n]\setminus \{i^*\}$, all terms in the denominator in the definition of $\beta_i$ in Equation (\ref{eq:burning-prob}) are clearly positive. The numerator $1-\exp\left(-1+\frac{b_{\max}}{m}\right)$ is positive as well. Hence, $\beta<1$. By applying Lemma~\ref{lem:alg-2-exp-bound} with $S=[n]\setminus \{i,i^*\}$, we get
    \begin{align}\nonumber
        &p_i \cdot \int_0^1{\prod_{j\in [n]\setminus \{i,i^*\}}{\left(1-\frac{b_j\cdot p_j}{m}\cdot y\right)\ud{y}}} \geq p_i \cdot \int_0^1{\left(1-\frac{y}{m}\right)^{\sum_{j\in [n]\setminus \{i,i^*\}}{b_j}}\ud{y}}\\\nonumber
        &= \left(1-\frac{b_i-1}{3m}\right)\cdot \frac{m}{m-b_i-b_{i^*}+1}\cdot \left(1-\left(1-\frac{1}{m}\right)^{m-b_i-b_{i^*}+1}\right)\\\nonumber
        &\geq \left(1-\frac{b_i-1}{3m}\right)\cdot \frac{m}{m-b_i-b_{i^*}+1}\cdot \left(1-\exp\left(-1+\frac{b_i-1}{m}+\frac{b_{i^*}}{m}\right)\right)\\\label{eq:pi-times-integral}
        &\geq \frac{m}{m-b_{i^*}}\cdot \left(1-\exp\left(-1+\frac{b_{i^*}}{m}\right)\right)= \frac{m}{m-b_{\max}}\cdot \left(1-\exp\left(-1+\frac{b_{\max}}{m}\right)\right)
    \end{align}
Due to Equation~(\ref{eq:pi-times-integral}), Equation~(\ref{eq:burning-prob}) yields $\beta_i\geq 0$. The second inequality follows since $(1-1/m)^m\leq e^{-1}$ and $b_i+b_{i^*}\leq m$. The third inequality follows by applying Lemma~\ref{lem:technical-bound-for-q_it-final} with $x=\frac{b_i-1}{m}$ and $z=\frac{b_{i^*}}{m}$. The last equality follows from the definition of agent $i^*$ as the agent with the highest quota.

We now show that $\sigma\in [0,1)$. Since $b_{\max}<m$, the denominator in the definition of $\sigma$ in Equation (\ref{eq:stealing-prob}) is clearly positive. For the numerator, the same inequality implies that 
\begin{align*}
    1-\left(2-\frac{b_{\max}}{m}\right)\cdot \exp\left(-1+\frac{b_{\max}}{m}\right) &< 1- \exp\left(-1+\frac{b_{\max}}{m}\right)
\end{align*}
and, hence, $\sigma<1$. Also, using the property $e^z\geq 1+z$, we obtain
\begin{align*}
    \exp\left(-1+\frac{b_{\max}}{m}\right) &= \left(\exp\left(1-\frac{b_{\max}}{m}\right)\right)^{-1} \leq \left(2-\frac{b_{\max}}{m}\right)^{-1}.
\end{align*}
Hence, the numerator in the definition of $\sigma$ in Equation~(\ref{eq:stealing-prob}) is non-negative, implying that $\sigma\geq 0$.
\end{proof}
Our next statement indicates that RSBS has optimal distortion and almost ideal distortion gap.

\begin{theorem}\label{thm:rsbs-distortion-gap}
    Mechanism RSBS has distortion at most $\frac{e}{e-1}$ and  distortion gap at most $1.0765$.
\end{theorem}

\begin{proof}
We begin by proving a distortion upper bound for mechanism RSBS as a function of the ratio $b_{\max}/m$. 

\begin{lemma}\label{lem:rsbs-distortion-per-instance}
    Mechanism RSBS has distortion at most $\left(1-\left(1-\frac{b_{\max}}{m}\right)\cdot \exp\left(-1+\frac{b_{\max}}{m}\right)\right)^{-1}$ on $\bb$-matching instances with $m$ items and $n$ agents with quota $b_i\geq 1$ for agent $i\in [n]$ such that $\sum_{i\in [n]}{b_i}=m$ and $b_{\max}=\max_{i\in [n]}{b_i}$.
\end{lemma}

\begin{proof}
As mechanism RSBS is a refinement of mechanism RS, we will prove the lemma by adapting the main steps in the proof of Theorem~\ref{thm:alg-2}. Consider a $\bb$-matching instance with $m$ items, and $n$ agents with integer quota $b_i\geq 1$ for agent $i\in [n]$ such that $\sum_{i\in [n]}{b_i}=m$ and the underlying random valuation profile $V$ drawn from the profile of UF distributions $\mathcal{F}$. We will show that, for every agent $i\in [n]$ and $t\in [b_i]$, the probability $q_{it}$ that mechanism RSBS assigns the $t$-th top item of agent $i$ to her is at least $1-\left(1-\frac{b_{\max}}{m}\right)\cdot \exp\left(-1+\frac{b_{\max}}{m}\right)$. The lemma will then follow by Lemma~\ref{lem:template-upper}.

Let $i^*$ be the agent of maximum quota $b_{\max}$ that mechanism RSBS uses. Also, let $g$ be the random variable denoting the item that is the $t$-th top item of agent $i\in [n]\setminus \{i^*\}$ for some $t\in [b_i]$. For $j\in [n]\setminus \{i,i^*\}$, we denote by $X_j$ the binary random variable indicating whether the agent $j$ is a survivor and ranks the item $g$ in their top $b_j$ positions in their ranking ($X_j=1$) or not ($X_j=0$). Then, the probability $q^1_{it}$ that algorithm RSBS assigns item $g$ to agent $i$ at the end of phase 1 is
\begin{align}\label{eq:q_it-1-rsbs}
    q^1_{it} &= p_i\cdot \E\left[\frac{1}{1+\sum_{j\in [n]\setminus \{i,i^*\}}{X_j}}\right].
\end{align}
Then, item $g$ will be assigned to agent $i$ at the end of phase 2 if it was assigned to agent $i$ in phase 1 and was not burnt. This happens with probability 
\begin{align}\label{eq:q_it-2-rsbs}
    q^2_{it} &= (1-\beta_i)\cdot q^1_{it}.
\end{align}
Finally, item $g$ will be assigned to agent $i$ at the end of the execution of mechanism RSBS if it is assigned to agent $i$ by the end of phase 2 and is not stolen by agent $i^*$ during phase 3. Agent $i^*$ has item $g$ among her top $b_{i^*}=b_{\max}$ items with probability $b_{\max}/m$. Hence, the probability that item $g$ is not stolen by agent $i^*$ in phase 3 is $1-\sigma\cdot \frac{b_{max}}{m}$. Overall, using Equations~(\ref{eq:q_it-2-rsbs}) and~(\ref{eq:q_it-1-rsbs}), we obtain that the probability $q_{it}$ that item $g$ is assigned to agent $i$ by mechanism RSBS is 
\begin{align}\nonumber
    q_{it} &= \left(1-\sigma\cdot \frac{b_{\max}}{m}\right) \cdot q^2_{it}=\left(1-\sigma\cdot \frac{b_{\max}}{m}\right)\cdot \left(1-\beta_i\right)\cdot q^1_{it}\\\label{eq:rsbs-q_it}
    &=\left(1-\sigma\cdot \frac{b_{\max}}{m}\right)\cdot \left(1-\beta_i\right)\cdot p_i\cdot \E\left[\frac{1}{1+\sum_{j\in [n]\setminus \{i,i^*\}}{X_j}}\right].
\end{align}
We compute the expectation at the RHS of Equation (\ref{eq:rsbs-q_it}) by repeating the steps that yield derivation (\ref{eq:alg-2-q_it}) in the proof of Theorem~\ref{thm:alg-2} to get
\begin{align}\label{eq:rsbs-exp}
    \E\left[\frac{1}{1+\sum_{j\in [n]\setminus \{i,i^*\}}{X_j}}\right] &= \int_0^1{\prod_{j\in [n]\setminus \{i,i^*\}}{\left(1-\frac{b_j\cdot p_j}{m}\cdot y\right)}\ud{y}}.
\end{align}
Hence, using the definitions of probabilities $\beta_i$ and $\sigma$ from Equations (\ref{eq:burning-prob}) and (\ref{eq:stealing-prob}), respectively, and Equation (\ref{eq:rsbs-exp}), Equation (\ref{eq:rsbs-q_it}) yields
\begin{align*}
    q_{it} &= 1-\left(1-\frac{b_{\max}}{m}\right)\cdot \exp\left(-1+\frac{b_{\max}}{m}\right)
\end{align*}
for agent $i\in [n]\setminus \{i^*\}$, as desired.

Now, assume that item $g$ is among the top $b_{i^*}$ items of agent $i^*$. We denote by $\mathcal{E}$ the event that some agent different than $i^*$ has been assigned item $g$ at the end of phase 2 of mechanism RSBS. Let $\mathcal{E}_i$ be the event that agent $i\in [n]\setminus \{i^*\}$ has been assigned item $g$ at the end of phase 2. Notice that $\mathcal{E}$ is the disjoint union of events $\mathcal{E}_i$ for $i\in [n]\setminus \{i^*\}$. Furthermore, event $\mathcal{E}_i$ is true if item $g$ is the $t$-th top item of agent $i$ (this happens with probability $1/m$ for a given $t$) for $t\in [b_i]$ and, then, item $g$ is assigned by mechanism RSBS to agent $i$ at the end of phase $2$. Hence, 
\begin{align}\nonumber
    \Pr[\mathcal{E}] &= \sum_{i\in [n]\setminus \{i^*\}}{\Pr[\mathcal{E}_i]} = \sum_{i\in [n]\setminus \{i^*\}}{\sum_{t\in [b_i]}{\frac{1}{m}\cdot q^2_{it}}}\\\nonumber
    &=\sum_{i\in [n]\setminus \{i^*\}}{\sum_{t\in [b_i]}{\frac{1}{m}\cdot (1-\beta_i)\cdot p_i\cdot \E\left[\frac{1}{1+\sum_{j\in [n]\setminus \{i,i^*\}}{X_j}}\right]}} \\\label{eq:i-star-prob}
    &= \sum_{i\in [n]\setminus \{i^*\}}{\sum_{t\in [b_i]}{\frac{1-\exp\left(-1+\frac{b_{\max}}{m}\right)}{m-b_{\max}}}}=1-\exp\left(-1+\frac{b_{\max}}{m}\right).
\end{align}
The third equality follows by Equations (\ref{eq:q_it-2-rsbs}) and (\ref{eq:q_it-1-rsbs}), the fourth one by Equation (\ref{eq:rsbs-exp}) and the definition of probability $\beta_i$, and the last one since the double sum contains exactly $m-b_{i^*}=m-b_{\max}$ terms.

Now, using the definition of phase 3 of mechanism RSBS, we can easily compute the probability $q_{i^*t}$ that agent $i^*$ is assigned item $g$ by mechanism RSBS as
\begin{align*}
    q_{i^*t} &= 1-\Pr[\mathcal{E}]+\sigma\cdot \Pr[\mathcal{E}] = 1-\left(1-\frac{b_{\max}}{m}\right)\cdot \exp\left(-1+\frac{b_{\max}}{m}\right),
\end{align*}
completing the proof. The first equality follows since, in phase 3 of mechanism RSBS, agent $i^*$ is assigned item $g$ with certainty if no agent has gotten item $g$ at the end of phase 2 and steals it with probability $\sigma$ otherwise. The second equality follows by Equation~(\ref{eq:i-star-prob}) and the definition of probability $\sigma$.
\end{proof}

Now, notice that the function $1-(1-y)\cdot \exp(-1+y)$ has non-negative derivative $y\cdot \exp(-1+y)$ and is non-decreasing for $y\in [0,1]$. Hence, by Lemma~\ref{lem:rsbs-distortion-per-instance}, the distortion of mechanism RSBS is upper-bounded by $\left(1-1/e\right)^{-1}=\frac{e}{e-1}$, completing the proof of optimality for the distortion of RSBS.

The rest of the proof is devoted to proving the bound on the distortion gap. We will do so using Lemma~\ref{lem:rsbs-distortion-per-instance} together with the next technical lemma.

\begin{lemma}\label{lem:product-bound}
    Consider a $\bb$-matching instance with $m$ items and $n$ agents with agent $i\in [n]$ having integer quota $b_i\geq 1$ so that $\sum_{i\in [n]}{b_i}=m$ and let $b_{\max}=\max_{i\in [n]}{b_i}$. Then,
    \[\prod_{i\in [n]}{\left(1-\frac{b_i}{m}\right)}\geq \left(1-\frac{b_{\max}}{m}\right)^{\left\lfloor \frac{m}{b_{\max}}\right\rfloor}\cdot \left\lfloor \frac{m}{b_{\max}}\right\rfloor\cdot \frac{b_{\max}}{m}.\]
\end{lemma}

\begin{proof}
Let $x_1$, $x_2$, ..., $x_k$ be real numbers satisfying the constraint $\sum_{i\in [k]}{x_i}=1$ and let $x_{\max}=\max_{i\in [k]}{x_i}$. Then, the quantity $\prod_{i\in [k]}{\left(1-x_i\right)}$ is minimized when $k=\left\lceil 1/x_{\max}\right\rceil$, $\left\lfloor 1/x_{\max}\right\rfloor$ of the numbers have value $x_{\max}$ and if $x_{\max}$ does not divide $1$, one number has value $1-\left\lfloor \frac{1}{x_{\max}}\right\rfloor\cdot x_{\max}$. If $x_{\max}$ does not divide $1$, then
using the inequality $1-r\cdot z \geq (1-z)^r$ for $r,z\in [0,1]$ (with $z=x_{\max}$ and $r=\frac{1}{x_{\max}}-\left\lfloor \frac{1}{x_{\max}}\right\rfloor$), we get
\begin{align*}
    \left\lfloor \frac{1}{x_{\max}} \right\rfloor \cdot x_{\max} &=1-\left(\frac{1}{x_{\max}}-\left\lfloor \frac{1}{x_{\max}} \right\rfloor\right)\cdot x_{\max} \geq \left(1-x_{\max}\right)^{\frac{1}{x_{\max}}-\left\lfloor \frac{1}{x_{\max}}\right\rfloor}\, ,
\end{align*}
and hence,
\begin{align*}
    \prod_{i\in [k]}{\left(1-x_i\right)} &\geq \left(1-x_{\max}\right)^{\left\lfloor \frac{1}{x_{\max}}\right\rfloor}\cdot \left\lfloor \frac{1}{x_{\max}}\right\rfloor\cdot x_{\max}.
\end{align*}
Otherwise, if $x_{\max}$ divides $1$, we get 
\begin{align}\label{eq:x-inequality}
    \prod_{i\in [k]}{\left(1-x_i\right)} &\geq \left(1-x_{\max}\right)^{1/x_{\max}}=\left(1-x_{\max}\right)^{\left\lfloor \frac{1}{x_{\max}}\right\rfloor}\cdot \left\lfloor \frac{1}{x_{\max}}\right\rfloor\cdot x_{\max}
\end{align}
again. The lemma follows by setting $k=n$, $x_i=\frac{b_i}{m}$ for $i\in [n]$. Notice that this set of numbers satisfies the constraints.
\end{proof}
Denoting the RSBS mechanism by $\mathcal{M}$, using Lemmas~\ref{lem:rsbs-distortion-per-instance} and~\ref{lem:product-bound}, we obtain
\begin{align}\nonumber
    \text{distgap}(\mathcal{M}) &= \sup_{\bb}{\text{dist}(\mathcal{M},\bb)\cdot \left(1-\prod_{i=1}^n{\left(1-\frac{b_i}{m}\right)}\right)}\\\label{eq:distortion-gap-bound}
    &\leq \frac{1-\left(1-\frac{b_{\max}}{m}\right)^{\left\lfloor \frac{m}{b_{\max}}\right\rfloor}\cdot \left\lfloor \frac{m}{b_{\max}}\right\rfloor\cdot \frac{b_{\max}}{m}}{1-\left(1-\frac{b_{\max}}{m}\right)\cdot \exp\left(-1+\frac{b_{\max}}{m}\right)}.
\end{align}

\begin{figure}[ht]
\centering
\hspace*{-1cm} 
\begin{tikzpicture}[scale=0.9]
  \begin{axis}[
    width=11cm,
    height=7cm,
    xmin=0, xmax=1,
    ymin=0.99, ymax=1.08,
    samples=800,
    domain=0.001:0.999,
    xlabel={$b_{\max}/m$},
    ylabel={distortion gap},
    grid=both,
    thick
  ]
    \addplot[
      smooth
    ]
    {(1 - (1 - x)^(floor(1/x)) * floor(1/x) * x) /
     (1 - (1 - x) * exp(-1 + x))};
  \end{axis}
\end{tikzpicture}
\caption{A plot of the bound on the distortion gap of algorithm RSBS at the RHS of Equation (\ref{eq:distortion-gap-bound}) as a function of the ratio $b_{max}/m$.}\label{fig:plot}
\end{figure}
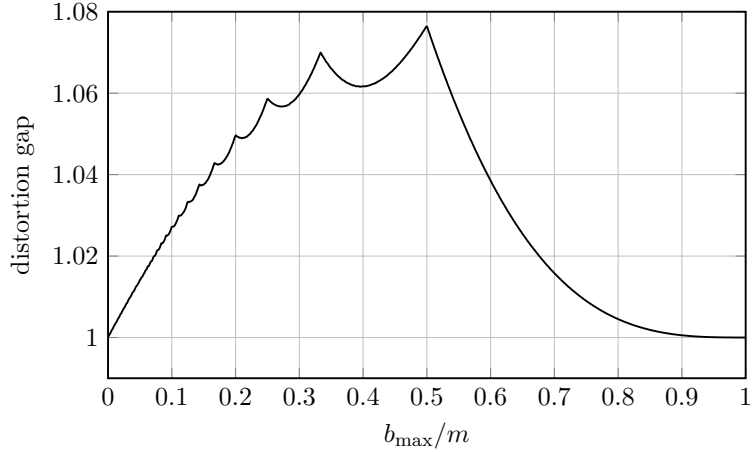
\noindent Figure~\ref{fig:plot} shows the value of the RHS of the above inequality as a function of $b_{\max}/m$. By inspecting the figure, it is clear that the function is maximized for $b_{\max}/m=1/2$ at the value $\frac{3}{4}\cdot \frac{2\sqrt{e}}{2\sqrt{e}-1}\leq 1.0765$, as desired.
\end{proof}

\section{Sequential One-Pass Mechanisms}\label{sec:sequential}
The non-trivial bound on the distortion gap of mechanism RSBS has been made possible using the additional ideas of burning and stealing items, which has resulted in a rather complicated mechanism. In this section, we focus on one-pass mechanisms that are considerably simpler. A one-pass matching mechanism considers the agents one by one according to a possibly random ordering. When considering an agent, the mechanism makes irrevocable decisions about the items that will be assigned to the agent.

Even though the class of one-pass ordinal mechanisms is broader, we focus here on two important subclasses. The first is the class of one-pass mechanisms that use a {\em fixed deterministic ordering} of the agents. The second one is the class of {\em secretary matching mechanisms}, which use a uniformly random agent ordering. We study them in Sections~\ref{subsec:fixed-ordering} and~\ref{subsec:secretary} below.

\subsection{Using a Deterministic Agent Ordering}\label{subsec:fixed-ordering}
For one-pass mechanisms that use a deterministic ordering of the agents, we prove tight bounds on their distortion and distortion gap. We begin with a lower bound on the distortion, which implies a lower bound on the distortion gap.

\begin{theorem}\label{thm:one-pass-det}
    Any one-pass ordinal matching mechanism has distortion at least $2-1/n$ on one-to-one matching instances with $n$ agents/items.
\end{theorem}

\begin{proof}
Consider a matching mechanism $\mathcal{M}$ applied on a one-to-one matching instance with $n$ agents/items. The agents are identified by their position in the deterministic ordering used by $\mathcal{M}$. We will show that there is some agent that has probability at most $\frac{n}{2n-1}$ to get her top item; the theorem will then follow by Lemma~\ref{lem:for-proving-lower-bounds}.

Assume that the expected number of items assigned to the agents in $[n-1]$ is $\beta(n-1)$. Then, there is an agent $i^*\in [n-1]$ that who gets her top item by mechanism $\mathcal{M}$ with probability at most $\beta$. Also, the probability that the top item of agent $n$ is available when it is considered by $\mathcal{M}$ is $1-\frac{\beta(n-1)}{n}$. Hence, some of the agents $i^*$ and $n$ has probability of getting their top item at most $\min\left\{\beta,1-\frac{\beta(n-1)}{n}\right\}\geq \frac{n}{2n-1}$.
\end{proof}

Since there are mechanisms with distortion $\frac{e}{e-1}$ on one-to-one matching instances, Theorem~\ref{thm:one-pass-det} implies the next lower bound on the distortion gap.
\begin{corollary}
    Any one-pass ordinal matching mechanism has distortion gap at least $2-2/e\approx 1.26424$.
\end{corollary}

In the rest of the section, we present and analyze the mechanism HQL (standing for highest quota last), which uses any ordering that has an agent with the highest quota last and achieves the optimal distortion and distortion gap guarantees for $\bb$-matching instances.

Specifically, HQL takes as input a $\bb$-matching instance with a set of items $m$ and a set of $n$ agents with integer quota $b_i\geq 1$ for agent $i\in [n]$ such that $\sum_{i\in [n]}{b_i}=m$ and $b_n\geq b_i$ for every $i\in [n]$. HQL considers the agents one by one. Whenever it considers agent $i$, it does the following with probability $p_i:=\frac{m}{2m-b_n-\sum_{j=1}^{i-1}{b_j}}$: it assigns to agent $i$ each of their $b_i$ top items that have not been assigned to agents before. HQL is depicted as Algorithm~\ref{alg:hal} below. Clearly, mechanism HQL is a one-pass mechanism. We begin with proving a distortion guarantee for it, which, by Theorem~\ref{thm:one-pass-det}, turns out to be the best possible among one-pass mechanisms that use a deterministic agent ordering.

\begin{algorithm}[h]
    \caption{The HQL Mechanism}
    \label{alg:hal}
    \begin{algorithmic}[1]
        \Input A \textbf{b}-matching instance with \(m\) items and \(n\) agents such that $b_n\geq b_i$ for $i\in [n]$. Query access to the set $\tp_i$ containing the favorite items of agent $i\in [n]$
        \Output A \textbf{b}-matching
        \State \(G \gets [m]\)
        \For{\(i \gets 1, ..., n\)}
            \State With probability \(\tfrac{m}{2m - b_n - \sum^{i-1}_{j = 1} b_j}: M_i \gets \tp_i \cap G; G\gets G\setminus M_i\)
        \EndFor
        \State\Return \(M\)
    \end{algorithmic}
\end{algorithm}

\begin{theorem}\label{thm:rank}
    Mechanism HQL has distortion at most $2-\frac{b_{\max}}{m}$ when applied on $\bb$-matching instances with $m$ items in which the maximum quota among all agents is $b_{\max}$.
\end{theorem}

\begin{proof}
    We will show that, for every $i\in [n]$ and $t\in [b_i]$ the probability that agent $i\in [n]$ is assigned their $t$-th best item $g$ is exactly $\frac{m}{2m-b_n}$. The theorem will then follow by Lemma~\ref{lem:template-upper}.

    Let $i\in [n]$ be an agent and $g$ an item that is among the top $b_i$ items of agent $i$. Item $g$ will be available when the mechanism considers agent $i$ if, for each agent $j=1, 2, ..., i-1$ considered before, the following event is true: either item $g$ is not among agent $j$'s top $b_j$ items or it is among agent $j$'s top $b_j$ items but is not assigned to agent $j$. For a given $j$, the probability that this event is true is equal to $1-\frac{p_j\cdot b_j}{m}$. Overall, for $i\in [n]$ we have that the probability that item $g$ is assigned to agent $i\in [n]$ by the mechanism is
    \begin{align*}
        p_i\cdot \prod_{j=1}^{i-1}{\left(1-\frac{p_j\cdot b_j}{m}\right)}.
    \end{align*}
    Notice that the probabilities $p_i$ are defined so that 
    \begin{align*}
        p_{i+1}\cdot \prod_{j=1}^i{\left(1-\frac{p_j\cdot b_j}{m}\right)} &= \frac{m}{2m-b_n-\sum_{j=1}^i{b_j}}\cdot \left(1-\frac{b_i}{2m-b_n-\sum_{j=1}^{i-1}{b_j}}\right) \cdot \prod_{j=1}^{i-1}{\left(1-\frac{p_j\cdot b_j}{m}\right)}\\
        &= \frac{m}{2m-b_n-\sum_{j=1}^{i-1}{b_j}}\cdot \prod_{j=1}^{i-1}{\left(1-\frac{p_j\cdot b_j}{m}\right)} = p_i \cdot \prod_{j=1}^{i-1}{\left(1-\frac{p_j\cdot b_j}{m}\right)},
    \end{align*}
    for every $i=1, 2, ..., n-1$. Hence, for every $i\in [n]$ and $t\in [b_i]$ the probability that agent $i\in [n]$ is assigned her $t$-th best item $g$ is equal to $p_1=\frac{m}{2m-b_n}$.
\end{proof}

We can now use the distortion bound from Theorem~\ref{thm:rank} and Corollary~\ref{cor:distortion-gap} to get a distortion gap bound which is also best possible among one-pass mechanisms that use a deterministic agent ordering.
\begin{theorem}\label{thm:rank-instance-optimal}
The distortion gap of mechanism HQL is at most $2-2/e\approx 1.26424$.
\end{theorem}

\begin{proof}
In the proof of Theorem~\ref{thm:rank-instance-optimal}, we will use Theorem~\ref{thm:rank} and the two technical lemmas below.
\begin{lemma}\label{lem:rank-product-bound}
    Consider a $\bb$-matching instance with $m$ items and $n$ agents with agent $i\in [n]$ having integer quota $b_i\geq 1$ so that $\sum_{i\in [n]}{b_i}=m$ and let $b_{\max}=\max_{i\in [n]}{b_i}$. Then, \[\prod_{i\in [n]}{\left(1-\frac{b_i}{m}\right)}\geq \left(1-\frac{b_{max}}{m}\right)^{m/b_{\max}}.\]
\end{lemma}

\begin{proof}
We prove the lemma by following very similar steps with those in the proof of Lemma~\ref{lem:product-bound}. Let $x_1$, $x_2$, ..., $x_k$ be real numbers satisfying the constraints $\sum_{i\in [k]}{x_i}=1$ and $0<x_i\leq x_{\max}\leq 1$ for $i\in [k]$. Then, the quantity $\prod_{i\in [k]}{\left(1-x_i\right)}$ is minimized when $k=\left\lceil 1/x_{\max}\right\rceil$, $\left\lfloor 1/x_{\max}\right\rfloor$ of the numbers have value $x_{\max}$ and if $x_{\max}$ does not divide $1$, one number has value $1-\left\lfloor \frac{1}{x_{\max}}\right\rfloor\cdot x_{\max}$. If $x_{\max}$ divides $1$, we obtain 
\begin{align}\label{eq:x-inequality-rank}
    \prod_{i\in [k]}{\left(1-x_i\right)} &\geq \left(1-x_{\max}\right)^{1/x_{\max}}.
\end{align}
Otherwise, using the inequality $1-r\cdot z \leq (1-z)^r$ for $r,z\in [0,1]$ (with $z=x_{\max}$ and $r=\frac{1}{x_{\max}}-\left\lfloor \frac{1}{x_{\max}}\right\rfloor$), we get
\begin{align*}
    \left\lfloor \frac{1}{x_{\max}} \right\rfloor \cdot x_{\max} &=1-\left(\frac{1}{x_{\max}}-\left\lfloor \frac{1}{x_{\max}} \right\rfloor\right)\cdot x_{\max} \geq \left(1-x_{\max}\right)^{\frac{1}{x_{\max}}-\left\lfloor \frac{1}{x_{\max}}\right\rfloor}\, ,
\end{align*}
and hence,
\begin{align*}
    \prod_{i\in [k]}{\left(1-x_i\right)} &\geq \left(1-x_{\max}\right)^{\left\lfloor \frac{1}{x_{\max}}\right\rfloor}\cdot \left\lfloor \frac{1}{x_{\max}}\right\rfloor\cdot x_{\max} \geq \left(1-x_{\max}\right)^{1/x_{\max}}.
\end{align*}
The lemma follows by setting $k=n$, $x_i=\frac{b_i}{m}$ for $i\in [n]$, and $x_{\max}=\frac{b_{\max}}{m}$. Notice that this set of numbers satisfies the constraints.
\end{proof}

\begin{lemma}\label{lem:rank-technical-lemma}
    For every $y\in (0,1]$, it holds $(2-y)\cdot \left(1-(1-y)^{1/y}\right)\leq 2-2/e$.
\end{lemma}

\begin{proof}
    We will first lower-bound $(1-y)^{1/y}$ by a polynomial in terms of $y$. Using the equivalent Taylor series for $\ln{(1-y})=-\sum_{j=1}^{\infty}{\frac{y^j}{j}}$, we have
\begin{align}\nonumber
    (1-y)^{1/y} &= \exp\left(\frac{\ln{(1-y)}}{y}\right) = \exp\left(-\frac{1}{y}\cdot \sum_{j=1}^{\infty}{\frac{y^j}{j}}\right) = \exp\left(-1-\sum_{j=1}^{\infty}{\frac{y^j}{j+1}}\right)\\\label{eq:polynomial}
    &\geq \exp\left(-1+\frac{y}{2}-\sum_{j=1}^{\infty}{\frac{y^j}{j}}\right)=\frac{1-y}{e}\cdot e^{y/2} \geq \frac{(1-y)\cdot (2+y)}{2e}
\end{align}
for $y\in (0,1]$. The last inequality follows since $e^{z}\geq 1+z$. Hence, 
\begin{align*}
    (2-y)\cdot \left(1-(1-y)^{1/y}\right) &\leq (2-y)\cdot \left(1-\frac{(1-y)\cdot (2+y)}{2e}\right)\\
    &= \frac{1}{2e}\left(4(e-1)-2(e-1)y+y^2-y^3\right) \leq 2-2/e.
\end{align*}
The last inequality follows since the function $4(e-1)-2(e-1)y+y^2-y^3$ has derivative $-2(e-1)+2y-3y^2$ which is negative for all $y\in [0,1]$.
\end{proof}
Now, using Corollary~\ref{cor:distortion-gap}, Theorem~\ref{thm:rank}, Lemma~\ref{lem:rank-product-bound}, and Lemma~\ref{lem:rank-technical-lemma} with $y=\frac{b_{\max}}{m}$, we obtain that the distortion gap of mechhanism HQL is
\begin{align*}
\text{distgap}(\text{HQL}) &\leq \sup_{\bb}{\text{dist}(\text{HQL},\bb)}\cdot \left(1-\prod_{i=1}^n{\left(1-\frac{b_i}{m}\right)}\right)\\
&\leq \left(2-\frac{b_{\max}}{m}\right)\cdot \left(1-\left(1-\frac{b_{\max}}{m}\right)^{m/b_{\max}}\right) \leq 2-2/e\,,
\end{align*}
as desired.
\end{proof}

\subsection{Secretary Matching Mechanisms}\label{subsec:secretary}
We now consider secretary mechanisms. Actually, a variant of mechanism RS has an implementation as a secretary mechanism. The mechanism first declares each agent as surviving in the same way mechanism RS does. Then, it considers the agents in a uniformly random order. Whenever considering agent $i$, the mechanism does nothing if the agent is non-surviving. Otherwise, it assigns to her any of her favorite items that have not been previously assigned to any agent. In this way, each item is given to an agent selected uniformly at random among the surviving agents who have this item as favorite. The main difference with the original mechanism RS is that the assignment of different items is not performed independently. This does not affect the analysis since the main argument used in the proof of Theorem~\ref{thm:alg-2} to account for the contribution of the items to the social welfare of the matching returned by the mechanism is linearity of expectation. Thus, the distortion bound of $\frac{e}{e-1}$ carries over to the secretary variant of the mechanism RS.

We are not aware of any secretary mechanism with a distortion gap below $\frac{e}{e-1}$. We can show, however, that near-optimal distortion gaps are not possible.

\begin{theorem}\label{thm:secretary-lower-bound}
    Any secretary ordinal matching mechanism has a distortion gap of at least $4/3$.
\end{theorem}

\begin{proof}
    We consider a $\bb$-matching instance with $m$ items and two agents with quotas $b_1=m-1$ and $b_2=1$. We will first show that the distortion of any secretary mechanism on this instance is at least $\frac{4m-2}{3m-1}$. We do so by showing that either the expected number of favorite items that agent 1 gets is at most $\frac{3m-1}{4m-2}\cdot (m-1)$ or the probability that agent 2 gets their top item is at most $\frac{3m-1}{4m-2}$. Then, the distortion lower bound follows by Lemma~\ref{lem:for-proving-lower-bounds} for $\rho=\frac{3m-1}{4m-2}$. 

    Assume that agent 1 receives an expected number of $\beta(m-1)$ favorite items when considered first. Then, the probability that agent 2 gets her top item is (at most) $1-\frac{\beta(m-1)}{m}$. Overall, the expected number of favorite items assigned to agent 1 is at most $\frac{1}{2}\cdot (m-1)+\frac{1}{2}\cdot \beta(m-1)=\frac{1+\beta}{2}\cdot (m-1)$ while the probability that agent 2 gets her top item is $\frac{1}{2}\cdot 1 +\frac{1}{2}\cdot \left(1-\frac{\beta(m-1)}{m}\right)=1-\frac{\beta(m-1)}{2m}$. Hence, if $\beta \leq \frac{m}{2m-1}$, agent 1 gets at most $\frac{3m-1}{4m-2}\cdot (m-1)$ favorite items on average while if $\beta\geq \frac{m}{2m-1}$, agent 2 gets her top item with probability at most $\frac{3m-1}{4m-2}$.   

    A distortion of at least $\frac{m-1}{m-2}$ is obtained by a mechanism that first assigns agent 2 her top item and then assigns (at least) $m-2$ available favorite items to agent 1. Therefore, the distortion gap is lower-bounded by $\frac{4m-2}{3m-1}\cdot \frac{m-2}{m-1}$, which approaches $4/3$ as $m$ goes to infinity.
\end{proof}

\section{Extensions and Open Problems}\label{sec:conclusion}

Throughout this paper, our focus has been on the impact of limited access to the agents' preferences without worrying about whether the agents could strategically misreport the information requested from them (e.g., the list of their most-preferred $b_i$ items). However, it is worth noting that even if the agents are strategic, our proposed mechanisms satisfy Bayesian incentive compatibility, i.e., it is a Bayesian Nash equilibrium for every agent to report the truth. This result is similar in spirit to the results of~\citet{GPTV24} and provides stronger support for the practicality of these mechanisms.

\begin{observation}\label{obs:BIC}
All the mechanisms that we introduced are Bayesian Incentive Compatible (BIC) for any UF distribution profile ${\mathcal{F}}$.
\end{observation}

Another assumption that we have been making throughout this paper is that the agents' valuations are additive across items. However, all our upper bounds on the distortion and the distortion gap of different mechanisms hold even for the more general class of \emph{submodular} valuation functions. Submodular valuations allow for some items to be substitutes and satisfies $v(S\cup \{j\})-v(S)\geq v(T\cup \{j\})-v(T)$ for any sets of items $S\subseteq T\subset [m]$ $v(S)$ and any item $j\in [m] \setminus T$ that belongs to neither one. 
\begin{observation}\label{obs:submod}
All our distortion and distortion gap bounds hold even if the stochastically generated preferences of each agent $i$ are captured by a submodular valuation function $v_i(\cdot)$, as long as this satisfies the UF property, i.e., each bundle $S$ of $b_i$ items is equally likely to be the one maximizing $v_i(S)$.
\end{observation}

We briefly justify both of the observations above in Appendix~\ref{sec:observ}.

In terms of interesting problems, the most intriguing one is whether a distortion gap of $1$ is possible by some relatively simple (and, ideally, one-pass) mechanism. 
An intermediate challenging question is whether our distortion lower bound of $\left(1-\prod_{i\in [n]}{(1-b_i/m)}\right)^{-1}$ is tight for every $\bb$-matching instance. Another interesting direction would be to extend our setting by allowing agents draw their valuations from non-UF distribution profiles. On the other hand, it is impossible for any prior-independent mechanism to achieve any non-trivial results if each value $v_{i,j}$ is drawn from a distribution $F_{i,j}$ that is both agent-specific and item-specific (this would allow for all the value to be concentrated in a single $v_{ij}$ value that the mechanism cannot detect). Are there interesting intermediate settings where low distortion bounds can be achieved via prior-independent mechanisms?

\section*{Acknowledgements}
Caragiannis and Homrighausen were supported by Independent Research Fund Denmark (DFF) under grant 2032-00185B. Gkatzelis was supported by NSF CAREER award CCF-2047907 and NSF grant CCF-2210502.

\bibliographystyle{plainnat}
\bibliography{references,references2}
\appendix

\newpage
\section{Observations from Section~\ref{sec:conclusion}}\label{sec:observ}

We provide proof sketches for Observations~\ref{obs:BIC} and \ref{obs:submod} from Section~\ref{sec:conclusion}.

\subsection*{Proof Sketch of Observation~\ref{obs:BIC} (Bayesian Incentive Compatibility)}
To verify that our mechanisms are all Bayesian incentive compatible note that the input they take from each agent $i$ is just the (unordered) set $S_i$ of their top $b_i$ items. Also, note that for each agent $i$, it is equally likely for their favorite bundle $S_i$ to be any bundle of $b_i$ items, so the demand for all items is ex-ante symmetric. Given this symmetry, each of our mechanisms yields a probability $q_{ij}$ with which agent $i$ receives each item $j\in S_i$, which is the same for all $j\in S_i$ (the allocation is determined only by the bundles reported by the agents). Similarly, the probability $q_{ij}$ that $i$ receives an item $j\notin S_i$ is the same for all $j\notin S_i$. Therefore, all that we need to observe is that $i$'s probability of receiving one of their favorite items (one in $S_i$) is higher than the probability of receiving any other item (not in $S_i$). This is easy to verify, since our mechanisms allocate all the ``favorite'' items first and then arbitrarily allocate the remaining items (in fact, our distortion and distortion gap bounds hold even if the mechanisms allocate none of the non-favorite items).

\subsection*{Proof Sketch of Observation~\ref{obs:submod} (Extention to Submodular Valuations)}
To verify the fact that our distortion and distortion gap upper bounds extend to submodular valuations, we adapt the proof of Lemma~\ref{lem:template-upper} as follows:

\begin{lemma}
Consider a matching mechanism $\mathcal{M}$ applied on a $\bb$-matching instance with an underlying profile of submodular valuation functions $(v_i)_{i\in [n]}$ drawn from UF distribution $\mathcal{F}$. For $i\in [n]$ and $t\in [b_i]$, let $q_{ij}$ be the probability that $\mathcal{M}$ matches agent $i$ with an item $j$ in their most valuable bundle, $S^*_i$, and let $\rho=\min_{i\in [n], j\in S^*_i}{q_{ij}}$. Then, algorithm $\mathcal{M}$ has distortion at most $1/\rho$.
\end{lemma}
\begin{proof}
Consider a $\bb$-matching instance with $m$ items and $n$ agents with quota $b_i\geq 1$ for each agent $i\in [n]$ such that $\sum_{i\in [n]}{b_i}=m$ and a randomly generated submodular valuation function $v_i(\cdot)$ such that any bundle $S$ of $b_i$ items is equally likely to maximize $v_i(S)$. Let $S^*_i = \arg\max_{S \subseteq [m]: |S|=b_i}\{v_i(S)\}$ be the favorite bundle of $b_i$ items for agent $i$ according to this valuation. Let $q_{iS}$ denote the probability that each bundle $S\subseteq S^*_i$ is assigned to agent $i$.

Consider any mechanism $\mathcal{M}$ such that for every $j\in S^*$ we have $\sum_{j \ni S\subseteq S^*}q_{iS} \geq \rho$. Then, considering the value each agent gets from their favorite items that mechanism $\mathcal{M}$ assigns to them, the expected social welfare of the matching computed by algorithm $\mathcal{M}$ is
\begin{align*}
    \E_{V\sim \mathcal{F}}\left[SW(\mathcal{M}(P(V),\bb),V)\right] &\geq \E\left[\sum_{i\in [n]}{\sum_{S \subseteq S^*_i}{q_{iS}\cdot v_i(S)}}\right] \\
    &\geq  \rho \cdot \E\left[\sum_{i\in [n]} v_i(S^*_i)\right]\\
    &\geq \rho \cdot \E_{V\sim \mathcal{F}}[OPT(V,\bb)].
\end{align*}
The second inequality in the above derivation follows by standard properties of submodularity and its multilinear extension~\citep{CCPV11} and from the definition of $\rho$. The third one since no $\bb$-matching
can extract more value than $v_i(S^*_i)$ for agent $i$. 
\end{proof}

\end{document}